%% file: paper.tex
\documentclass[conference,compsoc]{IEEEtran}
\ifCLASSOPTIONcompsoc
  \usepackage[nocompress]{cite}
\else
  \usepackage{cite}
\fi

\input{preamble}

\newtheorem{theorem}{Theorem}[section]
\newtheorem{definition}{Definition}[section]

\usepackage{titlesec}
\titleformat*{\section}{\large\bfseries}
\titleformat*{\subsection}{\normalsize\bfseries}

\newcommand{\oursystem}{DPolicy\xspace}

\begin{document}

\title{
\large \bf \oursystem: Managing Privacy Risks Across Multiple Releases with Differential Privacy}

\author{{\rm Nicolas K\"uchler\textsuperscript{1}, Alexander Viand\textsuperscript{2}, Hidde Lycklama\textsuperscript{1}, Anwar Hithnawi\textsuperscript{3}}  \\
\\
{\textsuperscript{1}\textit{ETH Zurich} \ \textsuperscript{2}\textit{Intel Labs} \ \textsuperscript{3}\textit{University of Toronto}}}

\date{}

\input{sections/mycommands}

\maketitle

\begin{abstract}
\input{sections/abstract}

\end{abstract}

\input{sections/introduction}

\input{sections/background}
\input{sections/relatedwork}

\input{sections/design}

\input{sections/enforcement}

\input{sections/evaluation}

\input{sections/conclusion}

\vspace{-3pt}
\section*{Acknowledgments}
\vspace{-3pt}
\noindent
We thank the reviewers for their feedback and our sponsors for their generous support, including Meta, Google, and SNSF through an Ambizione Grant No. PZ00P2\_186050.

\bibliographystyle{plain}
\interlinepenalty=10000 %
\bibliography{references, referencesplus}

\appendices %
\begin{appendices}

\input{sections/appendix_ieee}
\end{appendices}

\section{Meta-Review}

The following meta-review was prepared by the program committee for the 2025
IEEE Symposium on Security and Privacy (S\&P) as part of the review process as
detailed in the call for papers.

\subsection{Summary}
This paper introduces DPolicy, a privacy risk management framework that applies differential privacy (DP) at an organizational scale to address cumulative privacy risks.
It introduces a DP Policy Language to define privacy semantics, an optimized Policy Enforcement mechanism for scalability, and integration with existing DP systems to manage privacy budgets effectively.
By considering the scope and context of data releases, DPolicy enhances privacy risk assessment across various data-sharing scenarios, including machine learning and public data releases.

\subsection{Scientific Contributions}
\begin{itemize}
  \item Creates a New Tool to Enable Future Science
  \item Provides a Valuable Step Forward in an Established Field
\end{itemize}

\subsection{Reasons for Acceptance}
Overall, the reviewers were positive of the paper, appreciating its focus on ensuring privacy across multiple data releases with related fields and its development of a framework based on predicate logic for composing differential privacy guarantees.

\end{document}

%% file: preamble.tex
\usepackage{graphicx}

\usepackage{soul}
\usepackage{listings}

\usepackage{mfirstuc}

\usepackage{algorithm}
\usepackage[noend]{algpseudocode}

\usepackage{tabularx}
\usepackage{booktabs}
\usepackage{etoolbox}
\usepackage[nointegrals]{wasysym}

\usepackage{xurl}

\usepackage{glossaries}
\glsdisablehyper

\usepackage{enumitem,amssymb}

\newlist{todolist}{itemize}{2}
\setlist[todolist]{label=$\square$}
\usepackage{pifont}

\usepackage[makeroom]{cancel}

\usepackage{subcaption}
\usepackage{comment}
\usepackage{xspace}

\usepackage{multirow}
\usepackage{soul}
\usepackage{tikz}

\usepackage{url}
\usepackage{hyperref}
\usepackage{cleveref}

\usepackage{amsthm}
\usepackage{amsfonts}
\usepackage{mathtools}
\usepackage[utf8]{inputenc}

\usepackage{bm}

\usepackage{color}
\usepackage{xcolor}

\usepackage{enumitem}

\usepackage{colortbl}
\definecolor{Gray}{gray}{0.65}
\definecolor{LightGray}{gray}{0.9}

\usepackage{array}
\usepackage{arydshln}

%% file: sections/mycommands.tex
\providecommand{\cfref}[1]{c.f.~\S\ref{#1}}

\definecolor{cb-black}      {RGB}{  0,   0,   0}
\definecolor{cb-blue-green} {RGB}{  0,  073,  073}
\definecolor{cb-green-sea}  {RGB}{  0, 146, 146}
\definecolor{cb-rose}       {RGB}{255, 109, 182}
\definecolor{cb-salmon-pink}{RGB}{255, 182, 119}
\definecolor{cb-purple}     {RGB}{ 73,   0, 146}
\definecolor{cb-blue}       {RGB}{ 0, 109, 219}
\definecolor{cb-lilac}      {RGB}{182, 109, 255}
\definecolor{cb-blue-sky}   {RGB}{109, 182, 255}
\definecolor{cb-blue-light} {RGB}{182, 219, 255}
\definecolor{cb-burgundy}   {RGB}{146,   0,   0}
\definecolor{cb-brown}      {RGB}{146,  73,   0}
\definecolor{cb-clay}       {RGB}{219, 209,   0}
\definecolor{cb-green-lime} {RGB}{ 36, 255,  36}

\providecommand{\secspacingtop}{}%
\providecommand{\secspacingbot}{}%
\providecommand{\subsecspacingtop}{}%
\providecommand{\subsecspacingbot}{}%

\providecommand{\sectionpage}{}

\providecommand{\sectionnumber}{\S}
\providecommand{\Mod}[1]{\ \mathrm{mod}\ #1}
\providecommand{\fakeparagraph}[1]{\vskip 0pt\noindent\textbf{#1.}}

\providecommand{\githuburl}{}

\newglossaryentry{policy}{name={policy}, plural={policies}, description={The high level objects that an admin configures.}}
\newglossaryentry{rule}{name={rule}, plural={rules}, description={The concrete constraints derived from the set of policies.}}

\newglossaryentry{extension}{name={extension}, plural={extensions}, description={Extension policies, define a set of extensions.}}

\newglossaryentry{category}{name={category}, plural={categories}, description={Semantic groups of attributes for which we establish budget constraints.}}
\newglossaryentry{level}{name={membership level}, plural={membership levels}, description={The categories have different membership levels.}}

\newglossaryentry{admin}{name={admin}, plural={admins}, description={The person / role writing the policies.}}

\newglossaryentry{risklow}{name={low}, description={}}
\newglossaryentry{riskmedium}{name={medium}, description={}}
\newglossaryentry{riskhigh}{name={high}, description={}}
\newglossaryentry{level:member}{name={member}, description={}}
\newglossaryentry{level:strong}{name={strong connection}, description={}}
\newglossaryentry{level:weak}{name={weak connection}, description={}}

\newglossaryentry{trap}{name={scenario}, description={}}
\newglossaryentry{trap:context}{name={\ul{S1: Context}}, description={}}
\newglossaryentry{trap:scope}{name={\ul{S2: Scope}}, description={}}
\newglossaryentry{trap:time}{name={\ul{S3: (Time-based) Privacy Units}}, description={}}

\newacronym{ml}{ML}{machine learning}

\newacronym{abac}{ABAC}{Attribute-based Access Control}
\newacronym{acl}{ACL}{Access Control List}

\newacronym{dp}{DP}{Differential Privacy}
\newglossaryentry{adp}{name={$(\epsilon, \delta)$ - DP}, description={}}
\newacronym{puredp}{$\epsilon$-DP}{$\epsilon$-DP}

\newacronym{rdp}{RDP}{Rényi DP}
\newacronym{zcdp}{zCDP}{Zero-Concentrated DP}
\newacronym{fdp}{$f$-DP}{$f$-DP}

\newacronym[plural=PAs,firstplural=partitioning attributes (PAs)]{pa}{PA}{partitioning attribute}

\providecommand{\myabbreviation}{Replacement of myabbreviation \xspace}

%% file: sections/abstract.tex
Differential Privacy (DP) has emerged as a robust framework for privacy-preserving data releases
and has been successfully applied in high-profile cases, such as the 2020 US Census.
However, in organizational settings, the use of DP remains largely confined to isolated data releases.
This approach restricts the potential of DP to serve as a framework for comprehensive privacy risk management at an organizational level.
Although one might expect that the cumulative privacy risk of isolated releases could be assessed using DP's compositional property,
in practice, individual DP guarantees are frequently tailored to specific releases,
making it difficult to reason about their interaction or combined impact.
At the same time, less tailored DP guarantees, which compose more easily, also offer only limited insight because they lead to excessively large privacy budgets that convey limited meaning.
To address these limitations, we present \oursystem, a system designed to manage cumulative privacy risks across multiple data releases using DP.
Unlike traditional approaches that treat each release in isolation or rely on a single (global) DP guarantee,
our system employs a flexible framework that considers multiple DP guarantees simultaneously, reflecting the diverse contexts and scopes typical of real-world DP deployments.
\oursystem introduces a high-level policy language to formalize privacy guarantees, making traditionally implicit assumptions on scopes and contexts explicit.
By deriving the DP guarantees required to enforce complex privacy semantics from these high-level policies, \oursystem enables fine-grained privacy risk management on an organizational scale.
We implement and evaluate \oursystem, demonstrating how it mitigates privacy risks that can emerge without comprehensive, organization-wide privacy risk management.

%% file: sections/introduction.tex
\subsecspacingtop
\section{Introduction}
\label{sec:intro}

Remarkable algorithmic advances in data analytics over the past decade, paralleled with an unprecedented capacity to capture and process large-scale data, have shed light on the impact that data-driven approaches can have in addressing a wide range of complex societal problems.
This progress has, in turn, fueled a surge in the adoption of data-driven analytics, reshaping virtually every sector of society.
This rapid adoption of data-driven applications did not materialize without issues and has amplified concerns about individual privacy.
In response, privacy regulations, such as the European General Data Protection Regulation (GDPR), the California Consumer Privacy Act (CCPA), and others, have been enacted to ensure organizations handle data responsibly.
In addition to this regulatory pressure, the rising liability associated with data misuse and breaches is prompting organizations to adopt more systematic approaches for managing the risks inherent in the collection and processing of personal data.
While end-to-end encryption and secure computation can address these issues in a variety of settings, they are not sufficient on their own in scenarios where data must be not only processed but also shared or released~\cite{Viand2021-sokfhe, Hastings2019-mpcsok}.
For these settings, privacy tools such as Differential Privacy have shown to be promising.

\gls{dp} is a mathematically rigorous framework designed to protect individual privacy in scenarios where data analysis results are shared or released~\cite{Dwork2006-originaldp}.
It is widely accepted today as the state-of-the-art solution for privacy-preserving data releases.
Its impact is already apparent through notable success in high-profile applications, including the US Census~\cite{Census2022-web}, COVID-19 reporting~\cite{Bavadekar2020-dpgooglecovid2, Aktay2020-dpgooglecovid3}, and numerous industrial use cases~\cite{Apple2017-dp, Adeleye2023-dpwikipedia, Rogers2021-dplinkedin, Xu2024-dpgboard}.
Regulatory bodies are increasingly endorsing DP in data-sharing practices, as seen in the EU Data Governance Act (DGA), highlighting its growing acceptance within the legal framework~\cite{EU2022-dga}.
This rapid adoption of \gls{dp} can be attributed to its ability to provide robust privacy protections for a wide variety of complex applications.
In contrast to prior ad-hoc approaches to privacy guarantees, \gls{dp} provides a formal framework for defining and quantifying privacy loss.
This allows \gls{dp} to provide strong guarantees that hold against all possible attack strategies, including attackers with auxiliary information.
In addition, \gls{dp} defines formal composition properties, which enable precise reasoning about the cumulative effects of multiple different releases.
While deriving \gls{dp} guarantees for applications can be highly involved, the resulting guarantees can be expressed in a small number of privacy parameters, independently of application details.

\gls{dp} was originally considered in the context of assessing
the privacy risk of participating in a release from the perspective of an individual.
However, \gls{dp} also holds significant potential as a tool for organizations to quantify and manage privacy risks at an organizational level.
By protecting the privacy of \emph{any} user, \gls{dp} naturally allows one to make statements about the privacy risks associated with performing a data release.
More importantly, due to its composition properties, \gls{dp} enables reasoning about the cumulative privacy risks of performing a release in the context of other releases.
Despite this, \gls{dp} has not yet been systematically used or integrated into organizations for this purpose.
There exists a significant gap between the conceptual potential of \gls{dp} as a comprehensive privacy risk management framework and the way \gls{dp} is being implemented and used in practice.
Currently, organizations primarily consider \gls{dp} from the perspective of individual releases, similar to how prior ad-hoc anonymization techniques were used.
However, this approach fails to capitalize on \gls{dp}’s potential for evaluating and assessing cumulative privacy risks.
Shifting to an organization-wide perspective requires a fundamentally different approach to \gls{dp} guarantees as, in practice, there are significant gaps between the isolated guarantees provided by current practices and the desired global guarantees.

\fakeparagraph{\gls{dp} at Organizational Scale}
In current organizational practices, differentially private data releases are generally reported in isolation,
describing the datasets, algorithms, and corresponding privacy guarantees.
However, this approach is only justified if the releases are independent, an assumption that may not always hold in practice.
If any of these implicit assumptions fail to hold, the actual privacy risk associated with the release may significantly exceed what the per-release privacy guarantees suggest.
For example, multiple releases over related datasets can jointly enable practical reconstruction attacks, even if each release's privacy guarantees are sufficiently strong to prevent such attacks when considered in isolation (\cfref{sec:bg}).
Although one might expect that, for a small number of releases, the cumulative privacy risk could be easily assessed under different independence assumptions due to the compositional property of \gls{dp}, in practice, the privacy guarantees reported for real-world data releases are often highly tailored to the specific use case.
For example, the recent Wikimedia statistics release~\cite{Adeleye2023-dpwikipedia,Desfontaines2022-dppractice} describes the privacy guarantees for viewing and editing statistics using several complex, state-of-the-art approaches.
In particular, they define the protected change (i.e., privacy unit) that is at the root of \gls{dp} guarantees in project- and release-specific ways, which vary across different data releases.
This is essential for achieving tight privacy bounds but complicates the task of interpreting or comparing privacy guarantees across releases. For instance, in the Wikimedia example, the releases on editing and viewing statistics apply different definitions of users and time scales and employ varied geographic groupings.
While each guarantee is interpretable and reasonable in isolation, these  differences make it challenging, if not impossible, to assess cumulative privacy risks across multiple releases.

Relying on complex privacy units for \gls{dp} analysis
not only complicates cumulative privacy impact assessments,
it can also create a false sense of privacy:
while narrow privacy units enable attractive privacy loss parameters,
the actual protection offered may be significantly weaker if any of their underlying assumptions fall short.
Even when all implicit independence assumptions hold and the individual privacy guarantees employ compatible privacy units, it can be challenging to obtain a useful understanding of the overall privacy implications solely from these guarantees.
Even with state-of-the-art composition theorems~\cite{Mironov2017-rdp, Bun2016-zCDP, Dong2019-fdpclass}, the cumulative effect of multiple data releases causes the privacy parameters (e.g., $\epsilon$) to increase rapidly to levels that are no longer meaningful on their own.
This is because the guarantees provided by \gls{dp} degrade exponentially with increasing $\epsilon$.
This severely limits the range of meaningful privacy parameters, for example, at $\epsilon \geq 7$, the privacy leakage for the worst-case attacker considered in \gls{dp} is already effectively no longer meaningfully bounded mathematically~\cite{Desfontaines2024-largeeps}.

For simple applications such as counting queries, it is easy to observe that, as the privacy parameters increase, the amount of noise required to satisfy the \gls{dp} definition becomes too small to hide individual contributions in practice.
Empirical privacy auditing has shown that this also holds for more complex applications such as \gls{ml}~\cite{Nasr2021-dpmladv}.
Specifically, in the worst-case scenario, there is no gap between the theoretical differential privacy upper bounds and the practical success rates of attacks against the DP-SGD algorithm~\cite{Nasr2021-dpmladv}.
However, depending on certain (reasonable) deployment assumptions, the currently best-known privacy attacks indicate that there may be a significant gap in more constrained settings, making larger privacy budgets justifiable and generally accepted in practice in these contexts~\cite{Nasr2021-dpmladv, Nasr2023-fdpmlaudit, Steinke2024-dpsdgheuaud}.
In real-world deployments, there is a variety of contexts that one might want to consider,
e.g., privacy guarantees appropriate for internal (exploratory) releases %
frequently differ from those for external public releases.
This poses a fundamental challenge in applying \gls{dp} for privacy risk management at an organizational scale, as privacy losses appropriate for different contexts cannot be directly combined:
privacy parameters appropriate for counting queries
are unlikely to be sufficient to allow meaningful \gls{ml} training.
Meanwhile, generally accepted privacy parameters for \gls{ml} training %
would be highly inappropriate if applied to a single counting query release. %
As a result, considering only a single global privacy guarantee across different contexts is generally not a viable approach to \gls{dp} at an organizational level.

\fakeparagraph{Our Approach}
In this work, we enable  \gls{dp} to deliver on its potential as a risk management tool for organizational-scale data operations.
The key insight of our approach is that, to provide meaningful  \gls{dp} guarantees at this scale, we are not limited to individual  \gls{dp} guarantees.
Instead, we can achieve this by leveraging the combined effect of a set of guarantees specifically designed to complement each other.
We make the typically implicit assumptions underlying state-of-the-art \gls{dp} analysis visible and manageable by explicitly defining the scope and context of each guarantee.
We refer to the combination of scope, context, and  \gls{dp} guarantee as a \emph{\gls{rule}}, and the overall privacy guarantee achieved by a set of such \emph{\glspl{rule}} as the \emph{privacy semantics}.
By considering sets of \glspl{rule} rather than individual guarantees, our approach enables effective management of different contexts within the same system.
For example, this flexibility allows \gls{ml} applications to use empirically motivated, higher privacy loss parameters, while simultaneously enforcing stricter constraints
for releases requiring a worst-case attacker model.
Similarly, we can simultaneously consider multiple scopes, such as varying independence assumptions, to express more nuanced privacy semantics that incorporate both privacy loss parameters and the risks arising from potentially flawed assumptions.
Our approach utilizes narrow scopes (that are meaningless by themselves) to control privacy risks at a more granular level, even down to specific data attributes.
Finally,
our approach can provide release-specific guarantees (i.e., as in current practice) while still maintaining comparability of different releases, by also considering less specialized, compatible privacy units. %
Beyond this, by viewing releases through the lens of multiple privacy units simultaneously, our approach provides combined guarantees beyond what is possible to obtain from any given privacy unit individually.

\fakeparagraph{\oursystem}
In this work, we introduce \oursystem, a privacy risk management system which presents three key contributions:

\noindent
\emph{(i)} \emph{\gls{dp} Policy Language:}
Capturing the desired overall privacy semantics can require a large number of different scopes and contexts and, therefore, result in a large number of \glspl{rule} to consider.
At the organizational level, manually specifying and managing such a rule set is impractical, as ensuring that the rule set is permissive enough to accommodate intended use cases without inadvertently creating gaps in privacy guarantees poses a significant challenge.
\oursystem addresses this by providing a high-level \gls{policy} language to describe the desired privacy semantics in a concise and interpretable manner.
Our system translates a small set of \glspl{policy} into the large and complex set of interrelated \glspl{rule} necessary to encode the specified privacy semantics.

\noindent
\emph{(ii)} \emph{Policy Enforcement:}
In \oursystem, achieving the privacy semantics defined by the \gls{policy} set requires checking each release against all \glspl{rule} in the generated \gls{rule} set.
Specifically, this requires considering the composition of all prior releases in a \gls{rule}'s scope, making enforcement inherently stateful.
As enforcement complexity scales with the size of the \gls{rule} set, which can be very large for complex \glspl{policy}, this can introduce scalability issues.
We introduce an optimization that exploits the significant potential to reduce the size of the rule set by considering that the privacy guarantee of a rule might already be implied by another rule.
Specifically, our optimization reduces the \gls{rule} set size while preserving its privacy semantics by carefully identifying and pruning \emph{non-constraining} \glspl{rule}.

\noindent
\emph{(iii)} \emph{Integration with existing \gls{dp} Systems:}
\oursystem can manage privacy risks
for large-scale, complex one-off data releases (e.g., Census releases), or integrate with existing \gls{dp} budget allocation systems.
To demonstrate its practical applicability and effectiveness,
we implemented \oursystem and integrated it with Cohere, a state-of-the-art system for allocating limited privacy budgets across applications~\cite{Kuchler2024-cohere}.
We make our implementation of \oursystem available as open source\!\footnote{\url{https://github.com/pps-lab/dpolicy}}
and
evaluate it on a series of workloads to demonstrate
that \oursystem effectively prevents privacy risks that can occur with simpler approaches to privacy management.

%% file: sections/background.tex
\subsecspacingtop
\section{Background}
\label{sec:bg}
Below, we provide a brief overview of \gls{dp}
and refer to Dwork et al.~\cite{Dwork2014-dpbook} for a comprehensive formal treatment. %

\fakeparagraph{Differential Privacy}
\gls{dp} is a mathematical definition of privacy in the context of statistical data releases.
Informally, the definition captures that the result of an analysis should stay approximately the same, independent of whether or not any one individual contributed their data.
More formally, a randomized algorithm $M$  is a $(\epsilon, \delta)$-differentially private mechanism~\cite{Dwork2014-dpbook} if, for any set of results $\mathcal{S} \subseteq Range(M)$ and any two neighboring datasets $D, D'$, it holds that:
\begin{equation*}
    Pr[M(D)\in \mathcal{S}] \leq exp(\epsilon) \cdot Pr[M(D') \in \mathcal{S}] + \delta.
\end{equation*}
The privacy guarantees are determined by the parameters $\epsilon > 0$ and $\delta \in [0, 1)$ that bound the privacy loss of the release and by the definition of neighboring datasets, which describes the protected change.
Neighboring datasets are expressed in terms of a distance metric between databases, which captures two aspects:
\emph{(i)} what is the protected \emph{privacy unit}, and \emph{(ii)} how can the privacy unit change.

There are two commonly used ways for defining the change in neighboring datasets:
In \emph{bounded \gls{dp}}, the dataset size is known, and neighboring datasets differ by the substitution of a single unit.
In \emph{unbounded \gls{dp}}, the database size is unknown, and a single unit is either removed or added (rather than substituted).
Note that a privacy guarantee for unbounded \gls{dp} implies a guarantee for bounded \gls{dp} by observing that one addition and one removal corresponds to a substitute-one operation, i.e., an application of group privacy of size two~\cite{Desfontaines2019-dpsok}.
However, in bounded \gls{dp}, algorithms can leverage the fixed dataset size to enable a more refined analysis, making the two definitions not directly comparable.
Moreover, there are different models of \gls{dp} differing in trust assumptions.
Two widely adopted models are the central model of \gls{dp}, in which the party computing the \gls{dp} mechanism has access to the raw underlying data, and the local model, where no such trusted party is required but achieving similar utility requires larger privacy parameters.

\fakeparagraph{Composition and Group Privacy}
\gls{dp} composition theorems provide a method for bounding the cumulative privacy cost of multiple data releases.
For instance, by sequential composition, a series of $\epsilon_i$ \gls{dp} mechanisms, satisfies $(\sum_i \epsilon_i)$-\gls{dp}.
By parallel composition, a series of $\epsilon_i$ \gls{dp} mechanisms that operate on disjoint parts of a database satisfies  $(\max_i \epsilon_i)$-\gls{dp}.
Beyond the \gls{adp} variant, there are additional \gls{dp} variants with more attractive composition properties such as \gls{rdp}~\cite{Mironov2017-rdp} and \gls{zcdp}~\cite{Bun2016-zCDP}, which bound the average privacy loss using the Rényi divergence,
or \glstext{fdp}~\cite{Dong2019-fdpclass}, which follows the hypothesis testing interpretation of \gls{dp}.
These definitions provide different tradeoffs between flexibility, tightness, and composition complexity.
Conversions between different variants are possible in some directions, and all can be converted back into an \gls{adp} guarantee.

\gls{dp} guarantees are fundamentally concerned with the privacy loss between neighboring datasets, i.e., differing in one privacy unit.
Group privacy extends this notion by allowing up to $k$ differences, where $k$ is the group size.
Any $(\epsilon, 0)$-DP mechanism, is $(k \cdot \epsilon, 0)$-DP~\cite{Dwork2014-dpbook}.
For other \gls{dp} variants, or when $\delta \neq 0$ in \gls{adp}, group privacy is more involved and is separate from k-fold composition~\cite{Mironov2017-rdp, Bun2016-zCDP, Dong2019-fdpclass}.
For example, the $k$-fold composition of a $\rho$-\gls{zcdp} mechanism results in $(k \cdot \rho)$-\gls{zcdp}, while $k$-group privacy for the same mechanism results in $(k^2 \cdot \rho)$-\gls{zcdp}.

\fakeparagraph{Privacy Units}
In \gls{dp}, the natural privacy unit is the user, meaning all records associated with an individual are protected.
However, in practice, the privacy unit is often relaxed,
which allows achieving smaller privacy parameters at the cost of providing protections only for certain changes\footnote{Using group privacy allows obtaining an upper bound on the privacy loss across multiple privacy units.} and frequently harder to interpret
privacy semantics~\cite{Amin2024-dppractical}.
Deployments that perform releases regularly, e.g., monthly, frequently resort to time-based privacy units, where they, e.g., protect all user contributions of a single month instead of all user contributions.
Implicitly, such deployments make the assumption that the individual releases are hard to link across these time boundaries.
Without resorting to time-based privacy units, continuously applying new mechanisms under a finite budget is generally only possible through user rotation, where users are retired and replaced with new users as their budget is depleted~\cite{Kuchler2024-cohere, Luo2021-privacysched}.
Another form of privacy unit relaxation can consider that only part of the data is protected.
Kifer et al. define Attribute-DP~\cite{Kifer2011-dpnolunch}, where the focus is on protecting only an individual attribute of a database.
Label-DP\cite{Chaudhuri2011-labeldp} is a variant of this, where only the label of a supervised machine learning task is protected.
In addition, there are generalizations of the notion of neighboring datasets that can model various approaches, including time- and attribute-based privacy units~\cite{He2013-blowfish, Kifer2014-pufferfish}.

\fakeparagraph{Privacy Attacks}
\gls{dp} guarantees hinge on privacy parameters, such as the choice of privacy unit and the value of $\epsilon$.
If these are not chosen appropriately, even \gls{dp}-protected releases can become vulnerable to privacy attacks.
Several works have demonstrated this phenomenon across various domains, including aggregate location data~\cite{Pyrgelis2017-dplocattack, Pyrgelis2020-dplocattack}, tabular data~\cite{Stadler2020-syntheticdatafail, Annamalai2024-dpsynthattack} and machine learning models~\cite{Jayaraman2019-dpmlattack, Nasr2021-dpmladv, Carlini2022-minf, Nasr2023-fdpmlaudit}.
Such vulnerabilities have also been demonstrated in real-world \gls{dp} deployments.
For instance, Gadotti et al. showed that repeated observations can be exploited to extract sensitive information from Apple's data collection~\cite{Gadotti2022-dpappleattack}.
Moreover, Houssiau et al. demonstrated that incorrect assumptions about the privacy unit led to a false sense of privacy in Google's \gls{dp}-based location data release~\cite{Houssiau2022-dplocfail}.
While most attack literature focuses on specific releases, in practice, the line between a single release with multiple queries and multiple releases is blurry, and most existing attack literature also directly applies to the multi-release setting  considered in this paper, as we illustrate below. %

First, we consider how  neglecting privacy across multiple releases can trivially empower attackers by demonstrating how  membership inference attacks can be easily combined across different machine learning models, even when the models operate in entirely different domains.
For example, the likelihood ratio attack (LiRA)~\cite{Carlini2022-minf} trains \emph{shadow} models on datasets that either include or exclude the target sample.
It then performs a likelihood ratio test between the target sample's loss on the actual model and the loss distributions generated by the two sets of shadow models.
In a multi-release setting, the confidence of such an attack can be significantly amplified if the adversary knows that an individual is either included in the training data of both models or neither.
For example, if independent attacks on each model result in a 75\% belief, Bayesian inference can combine these probabilities to increase confidence to 90\%.
Per-release management approaches are fundamentally incapable of addressing this, requiring an approach that can track risks across different releases.

Second, we consider how underlying user data can be recovered even from \gls{dp}-protected synthetic data releases unless privacy is considered at a fine-grained per-attribute level.
For example, Annamalai et al. demonstrate an attack on the ``select-measure-generate'' approach~\cite{Liu2021-dpsynthunify} used by most synthetic data generation algorithms~\cite{Annamalai2024-dpsynthattack}.
In each round (which we can consider a release), these algorithms
select a set of attributes and release their marginals to refine the synthetic representation.
By default, these algorithms do not provide a mechanism to control how often specific attributes are selected.
As a result, there is a risk that the target attribute in an attribute-inference attack may accumulate a disproportionate share of the overall privacy cost, leading to a successful attack.
Preventing such issues requires a substantially more fine-grained approach to privacy risk management, that can prevent the accumulation of privacy loss and ensure a more even distribution.

%% file: sections/relatedwork.tex
\section{Related Work}
\label{sec:relwork}

Below, we discuss work relevant to \oursystem, including existing \gls{dp} policy and management approaches.

\fakeparagraph{Semantic Structure of the Privacy Budget}
The most closely related line of work to \oursystem includes approaches that augment guarantees over large privacy parameters with additional fine-grained information.
Ghazi et al., define per-attribute Partial-DP~\cite{Ghazi2022-dpcontrol}, which enables them to assign a smaller per-attribute budget, $\epsilon_0$, while permitting a larger overall $\epsilon$.
Their approach justifies a higher total $\epsilon$ by showing that privacy loss is not disproportionately concentrated on any single attribute.
Our approach similarly ensures that high overall
$\epsilon$ is not concentrated on individual attributes,
but also supports a much wider set of additional budgets,
does not require specific algorithmic support and can support both bounded and unbounded \gls{dp}.
Kifer et al., conducted an ex-post case study on the 2020
Census redistricting data release to examine the attribute-based privacy semantics of the release in finer detail~\cite{Kifer2022-censussemantics}.
They define eight scenarios, each corresponding to a unique combination of attributes, and accumulate privacy loss from queries influenced by these attributes.
This approach is similar in spirit to our approach in \oursystem, but their work only considered a small set of manually selected ad-hoc scenarios and lacks a systematic method for defining scenarios.
In addition, they only support attribute-based constraints and do not consider policy definition or constraint enforcement.

\fakeparagraph{Policy Frameworks for DP}
Unlike common cryptographic parameters, \gls{dp} parameters are widely accepted to be highly context-dependent and cannot be chosen in an application-agnostic manner.
This has resulted in a wide proliferation of approaches to determining and reporting \gls{dp} guarantees.
To promote transparency, accountability, and comparability of these parameters, prominent \gls{dp} researchers have called for the creation of an ``epsilon registry''~\cite{Desfontaines2022-dppractice, Oblivious2024-epsregistry}.
Currently, most \gls{dp} deployments are documented in technical blogs or research papers.
These documents serve as informal \gls{dp} policies, which minimally specify critical aspects~\cite{Ponomareva2023-dpmlbestpract} such as the \gls{dp} model (e.g., central or local \gls{dp}), the neighboring relation (e.g., add/remove or replace-one), the privacy unit (e.g., user-level), and the privacy parameters (e.g., $\epsilon$, $\delta$).
In some cases, additional information is provided, including hyperparameter tuning, the type of accounting used, and specific mechanisms.
Such details enhance transparency and provide a richer context for evaluating \gls{dp} guarantees.
Standardization bodies, such as the National Institute of Standards and Technology (NIST), are also working to improve this area.
For example, NIST is developing guidelines for evaluating \gls{dp} guarantees~\cite{NIST2023-dpguideline}.
Most of this documentation covers individual deployments; hence, obtaining an aggregated view of privacy risk requires piecing together information from multiple sources.
Apple is a notable exception, as they provide a consolidated document listing all their \gls{dp} applications~\cite{Apple2017-dp}.
However, even in this case, the cumulative privacy impact across all deployments is not explicitly discussed and would require further analysis of the individual parameters.
\oursystem not only facilitates per-release \glspl{policy} but also provides structured, comprehensive global \glspl{policy}, allowing for a cohesive overview of aggregate \gls{dp} usage.
Beyond current \gls{dp} reporting practice, Benthall et al. consider an integration of \gls{dp} with the conceptual framework of contextual integrity (CI)~\cite{Benthall2024-dpcontextualint}.
This establishes a theoretical model considering \gls{dp}  based on the context in which data is used.
In their framework, CI views privacy as appropriate information flow within a specific social context.
Their approach augments CI with a descriptive transmission property that dictates that only information flows under a \gls{dp} guarantee are considered appropriate in certain contexts.
While this could be seen as a type of policy, their work is purely conceptual and orthogonal to the concrete challenges of providing meaningful guarantees over multiple releases.

\fakeparagraph{Privacy Management for DP}
 Recent research has focused on treating the \gls{dp} budget as a shared resource and consequently proposes the design of systems that manage access~\cite{Lecuyer2019-sage} or address the allocation of this resource to applications for optimized efficiency~\cite{Kuchler2024-cohere, Tholoniat2022-dppacking, Xiao2024-dpmgmt} and/or fairness~\cite{Luo2021-privacysched, Pujol2021-budgetsharing, Liu2024-dpmgmt}.
 For example, K\"uchler et al. propose Cohere, which focuses on unifying the allocation of a wide range of different \gls{dp} mechanisms. %
These approaches are orthogonal to our work, as \oursystem only provides the constraints and is independent of the concrete (manual or automated) allocation strategy used.

%% file: sections/design.tex
\subsecspacingtop
\section{Policy Language}
\subsecspacingbot
\label{sec:policy}

In this section, we introduce the policy language employed by \oursystem to specify desired privacy semantics.
Defining \gls{dp} \glspl{policy} enables a structured approach to managing privacy constraints without the need to specify an extremely large set of \glspl{rule} manually.
This approach mirrors the advantages of defining access control policies via \gls{abac} rather than managing an \gls{acl} directly.\footnote{Note that \oursystem is designed to complement, not replace, existing access control systems.}
In addition to providing a formal specification, we discuss how our policy language facilitates expressing complex privacy semantics and illustrate its use through concrete examples.

\begin{figure*}
    \centering
    \begin{subfigure}[t]{0.33\textwidth}
        \centering
          \includegraphics[width=\linewidth, trim=2 5 28 0, clip]{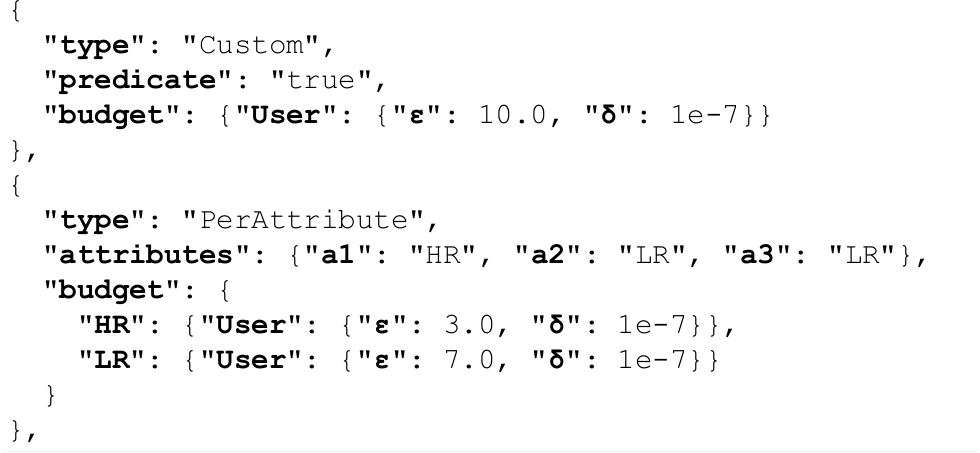}

        \vfill
        \caption{Base Policies [1/2]}
        \label{fig:policy:custom}
        \label{fig:policy:perattr}
    \end{subfigure}
    \hfill
    \begin{subfigure}[t]{0.31\textwidth}
        \centering
        \includegraphics[width=\linewidth, trim=2 5 70 0, clip]{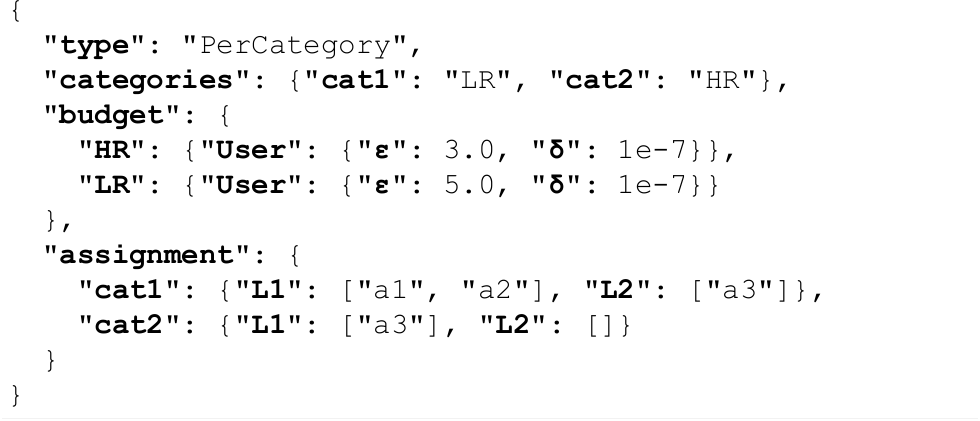}

        \vfill
        \caption{Base Policies [2/2]}
        \label{fig:policy:percat}
    \end{subfigure}
    \hfill
    \begin{subfigure}[t]{0.34\textwidth}
        \centering
        \includegraphics[width=\linewidth, trim=2 5 5 0, clip]{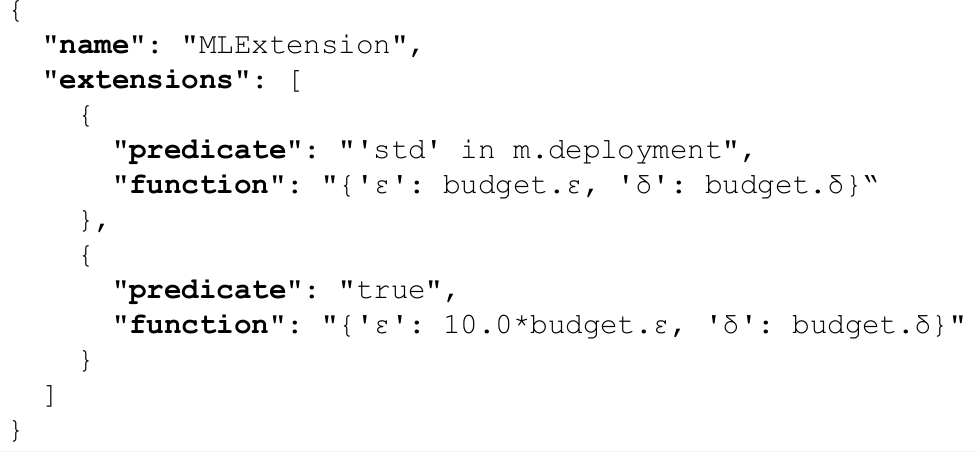}
          \vfill
        \caption{Extension Policy}
        \label{fig:policy:ext}
    \end{subfigure}
       \caption[An example \gls{policy} set in \oursystem]{An example \gls{policy} set in \oursystem, demonstrating a custom \gls{policy} defining a global user budget, per-attribute and per-category \glspl{policy} and an extension \gls{policy} for differentiating between standard and ``black-box'' ML contexts.
       }
       \label{fig:policies}
\end{figure*}

\subsection{Core Concepts}
\label{sec:policy:core}
At the core of \oursystem is a policy language that allows \glspl{admin} to specify the desired \gls{dp} guarantees through a set of high-level \emph{\glspl{policy}}.
From these policies, our system generates a set of concrete \emph{\glspl{rule}} that correspond to individual \gls{dp} guarantees.
A series of releases satisfies a set of \glspl{policy} only if it satisfies \emph{all} \glspl{rule} generated by those \glspl{policy}.
That is, \oursystem ensures that each individual guarantee holds independently of any other \glspl{rule} present in the system.
In addition to \glspl{policy} that apply across all releases, we also support \emph{per-release} policies.
These can model the type of isolated guarantees common in current approaches to \gls{dp} management.
\oursystem supports the same features for both per-release and across-release policies, so we focus on the latter in the following discussion.
\Glspl{policy} are divided into \emph{base} and \emph{extension} \glspl{policy}.
Base \glspl{policy} directly generate (intermediate) \glspl{rule}, while extension \glspl{policy} act on the \glspl{rule} generated by the base \glspl{policy} and expand each into multiple new \glspl{rule}.
\oursystem uses extension \glspl{policy} to efficiently express privacy guarantees in a context-specific manner.
In \Cref{sec:enf:prune}, we discuss how \oursystem applies a series of optimizations that eliminate redundant \glspl{rule} to minimize the size of the final set of \glspl{rule}.

\fakeparagraph{Formal Definition}
In \oursystem, the fundamental unit is a \emph{\gls{rule}}.
Each \gls{rule} defines a \gls{dp} budget (e.g., an $\epsilon$ limit) for a specific \emph{scope}, which determines which DP applications are covered by this guarantee, and a \emph{context}, which encodes assumptions about the setting of these \gls{dp} applications.
Formally, we consider a multimap of \emph{labels} on the  \gls{dp} mechanisms that make up the \gls{dp} applications.
Let $\mathcal{M}$ denote this universe of labeled \gls{dp} mechanisms.
We then model scopes and context as a predicate $\phi_k: \mathcal{M} \rightarrow \{0, 1\}$ over the mechanisms' labels.
A \gls{rule} $r_k$ is defined as a tuple $(\phi_k, U_k, B_k)$, consisting of such a predicate and a privacy budget $B_k$ relative to a privacy unit $U_k$.
Note that the budget can be expressed in any variant of DP (e.g., \gls{adp}, \gls{rdp}, \gls{zcdp}).
Given a composition of mechanisms $M = (m_1, m_2, \ldots, m_N)$, the rule $r_k$ requires that the (sub) composition of the mechanisms $M_k$ matching the predicate $\phi_k$, satisfies the privacy budget $B_k$\footnote{This is similar to applying a privacy filter only within the specific scope and context defined by the predicate.}.

In \oursystem, \glspl{rule} are derived from the set of policies provided by the \gls{admin}.
Let $\mathcal{R}$ denote the universe of \glspl{rule}.
There are two types of \glspl{policy} in \oursystem: \emph{(i)} base \glspl{policy}, and \emph{(ii)} extension \glspl{policy}.
Formally, each base \gls{policy} $P_B$ defines a set of intermediate rules:
\begin{equation}
    \hat{R} = \{(\phi_1, U_1, B_1), \; (\phi_2, U_2, B_2), \; \ldots \}
\end{equation}
An extension \gls{policy} $P_E$ with $\ell$ extensions defines a mapping $extend: \mathcal{R} \rightarrow \mathcal{R}^{\ell}$, which maps each intermediate \gls{rule} to a set of $\ell$ extended intermediate \glspl{rule} (\cfref{sec:policy:ctx}).
In \oursystem, the simplest type of \gls{policy} is a \emph{custom} base \gls{policy} that directly specifies the predicate (using CEL~\cite{Google2024-cel}), privacy unit, and budget (c.f.~\Cref{fig:policy:custom}).
For example, a \gls{dp} bound for a global scope can be expressed using a predicate $\phi_*(m) \mapsto 1$, which matches on all mechanisms.
However, our system primarily relies on significantly more powerful \glspl{policy} to concisely specify the intended privacy guarantees.
We discuss the policies and their use in the following.

\subsection{Context-specific Privacy Budgets}
\label{sec:policy:ctx}

Although \gls{dp}, in theory, provides a unified, application-independent approach to privacy, it is generally accepted that, in practice, adequate privacy budgets cannot be determined without additional context.
For example, privacy budgets suitable for counting queries (e.g., $\epsilon < 1$) are typically insufficient to allow meaningful \gls{ml} training.
Empirical privacy auditing has shown that \gls{dp} analyses of complex applications, such as the DP-SGD algorithm, are tight for worst-case attackers~\cite{Nasr2021-dpmladv}.  However, the best-known privacy attacks~\cite{Nasr2021-dpmladv, Nasr2023-fdpmlaudit, Steinke2024-dpsdgheuaud} suggest that a significant gap may exist for less powerful adversaries.
Therefore, depending on certain (reasonable) deployment assumptions, the use of larger privacy budgets for \gls{ml} might be justifiable and is generally accepted in practice~ \cite{Ponomareva2023-dpmlbestpract}.

\newglossaryentry{setting}{name={setting}, plural={settings}, description={}}

\fakeparagraph{Context-aware \Glspl{rule}}
Although the context-sensitive nature of \gls{dp} budgets might suggest that managing privacy risks in a unified, global manner is fundamentally impossible, in \oursystem we show how we can effectively model and address this complexity by leveraging context-aware \glspl{rule}.
Consider, for example, a scenario where we aim to enforce a strict privacy budget for general applications (``standard'' setting) while permitting a relaxed privacy budget in justified cases, such as in \gls{ml} deployments that limit an attacker's access capabilities (e.g., ``black-box'' setting).
An intuitive solution might involve defining a \gls{rule} with a relaxed privacy budget that matches only ``black-box'' ML deployments alongside a \gls{rule} with a strict privacy budget that matches any mechanism.
However, this approach does not succeed because the strict \gls{rule} would also apply to \gls{ml} mechanisms, thereby inadvertently enforcing the strict privacy budget on them.\footnote{Recall that in \oursystem, mechanisms must satisfy the privacy guarantees of all \glspl{rule} they match.}
Instead, we introduce a \gls{rule} with a strict budget that tracks \emph{only} applications that fall into the standard setting, and a second \gls{rule} with the relaxed budget that tracks the privacy costs of \emph{both} the standard and the black box \gls{ml} releases.
Note that, by themselves, neither rule provides sufficient guarantees:
the first rule omits many mechanisms from its accounting, and the second rule does not guarantee appropriate privacy loss for applications in the standard setting (e.g., count queries).
However, when these rules are employed together, they achieve the desired effect.
Black-box \gls{ml} applications can utilize relaxed budgets appropriate for their context, while all other applications remain constrained by the stricter budget.

\fakeparagraph{Extension \Glspl{policy}}
Manually specifying, for each (base) \gls{policy} in the system, how budgets change across different \glspl{setting} does not scale well in practice.
This is because it permits inconsistencies in both the assumptions underlying each setting and the degree of relaxation applied,  making it challenging to specify complex \gls{policy} sets.
To address this issue, \oursystem adopts a different approach.  As illustrated in \Cref{fig:policy:ext}, \glspl{admin} can create extension \glspl{policy} by defining various \glspl{setting} (e.g., ``standard'' and ``black-box'') and specifying how they affect the corresponding budgets.
This, however, raises the question of how to express this budget adjustment.
In \oursystem, we enable the \gls{admin} to provide a custom function that maps a base budget to a new budget.
To configure these functions, \glspl{admin} can leverage insights from empirical privacy auditing~\cite{Nasr2021-dpmladv}. %
Formally, an extension \gls{policy} $P_E$ is defined by a set of extensions:
\begin{equation*}
    E = \{(\phi_1, U_1, f_{1}), \; (\phi_2, U_2, f_2), \; \ldots, \; (\phi_{\ell}, U_{\ell}, f_{\ell})\}
\end{equation*}
Each extension is a pair $(\phi_k, U_k, f_k)$ consisting of a predicate over the mechanisms $\phi_k: \mathcal{M} \rightarrow \{0,1\}$, and a function $f_k: B \rightarrow B$ on privacy budgets for the privacy unit $U_k$.
For each (intermediate) rule in our system, an extension \gls{policy} with $\ell$ extensions generates a set with $\ell$ new rules by applying the $extend: \mathcal{R} \mapsto \mathcal{R}^{\ell}$ mapping:
\begin{equation*}
extend\left((\phi_k, U_k, B_{k})\right) \mapsto \left\{
\begin{aligned}
    & (\phi_1 \land \phi_k, \; U_k, \; f_{1}(B_{k})), \\
    & (\phi_2 \land \phi_k, \; U_k, \; f_2(B_{k})), \\
    & \qquad\qquad \cdots \\
    & (\phi_{\ell} \land \phi_k, \; U_k, \; f_{\ell}(B_{k}))
\end{aligned}
\right\}
\end{equation*}

Note that every extension \gls{policy} must contain an extension $(\phi_*, U_k, f_k)$ that applies to all mechanisms and bounds the total privacy budget, i.e., $\phi_* \mapsto 1$.
This is necessary to ensure that if a mechanism was part of the scope of the original intermediate rule, it is also guaranteed to still be in the scope of at least one of the new rules.
Without this requirement, applying an extension policy might have unintended consequences.
Similarly, each extension \gls{policy} typically contains an extension $(\phi_k, U_k, f_{id})$, where the function on privacy budgets is the identity function.
This ensures that the original budget specified in the intermediate rule is preserved in at least some setting.

\fakeparagraph{Combining Extension Policies}
In general, contexts in \gls{dp} depend not only on a single factor (e.g., \gls{ml} deployment assumptions) but on a combination of multiple factors (e.g., public release vs. internal exploratory use and/or current data vs. historical data).
For example, most applications require some form of \gls{dp} hyperparameter selection~\cite{Liu2019-dphyperselect, Papernot2021-dphyperparam}.
Accounting for this selection in \gls{dp} guarantees typically necessitates significantly larger privacy budgets, as worst-case attacks may exploit hyperparameter choices.
However, many reported \gls{dp} guarantees implicitly assume that such attacks are impractical and present their guarantees using the data custodian model, i.e., without accounting for hyperparameter tuning.
Similarly, it is often assumed that historical data is less privacy-sensitive, permitting higher privacy budgets.
While an extension \gls{policy} can represent arbitrarily complex contexts by considering multiple dimensions,
expressing this complexity within a single policy quickly becomes impractical.
Instead, in \oursystem, we support the automatic combination of different extension policies, each addressing a single dimension.
For example, one can define a \gls{policy} for \gls{ml} deployment budget relaxations and another \gls{policy} to differentiate between the budget allocated to the final released mechanisms and the budget including hyperparameter tuning.
This approach enables \oursystem to construct complex contexts concisely.

\begin{algorithm}[t]
\begin{algorithmic}[1]
\Function{ApplyExtensions}{irules, epolicies}
    \For{$\textsc{Extend}$ \textbf{in} epolicies}
        \State IR $\leftarrow$ $\emptyset$
        \For{irule \textbf{in} irules}
            \State IR $\leftarrow$ IR $\cup$  \Call{Extend}{irule}
        \EndFor
        \State irules $\leftarrow$ IR
    \EndFor
    \State \Return irules
\EndFunction
\end{algorithmic}
\caption{The \oursystem intermediate \gls{rule} expansion.}
\label{algo:iruleext}
\end{algorithm}

\Cref{algo:iruleext} demonstrates how \oursystem combines multiple extension \glspl{policy} to expand an initial intermediate set of \glspl{rule} (for the same privacy unit).
The algorithm sequentially applies the mapping of each extension \gls{policy} $extend: \mathcal{R} \mapsto \mathcal{R}^{\ell}$, replacing each intermediate \gls{rule}, with a set of newly generated intermediate \glspl{rule}.
Formally, let $\hat{R}$ be a set of base (intermediate) \glspl{rule}, and $(E_1, E_2, \ldots, E_M)$ represent a sequence of $M$ extension \glspl{policy}.
The resulting \gls{rule} set $R$ contains a (final) \gls{rule} for each element in the Cartesian product of $\hat{R}$ and the extensions, i.e., $|R| = |\hat{R}| \times |E_1| \times \ldots \times |E_M|$.
Each (final) \gls{rule} $(\phi_k, U_k, B_k) \in R$ contains a predicate:
\begin{equation}
\label{eq:rulepreddecomp}
     \phi_{k} = \phi_{k_B}^{(B)} \land \phi_{k_{1}}^{(1)} \land \ldots \land \phi_{k_{M}}^{(M)}
\end{equation}
which is the conjunction of the base \gls{rule} predicate $\phi_{k_B}^{(B)}$ and predicates from each extension \gls{policy}, and a budget:
\begin{equation}
     B_k = f_{k_M}^{(M)}\left( \; \ldots \; f_{k_2}^{(2)}\left(f_{k_1}^{(1)}\left(B_{k_B}^{(B)}\right)\right) \ldots \right)
\end{equation}
defined by applying the composition of the budget functions of the different extensions to the base budget $B_{k_B}^{(B)}$.

\subsection{Supporting Multiple Scopes}
\label{sec:policy:scope}
Even with state-of-the-art approaches to \gls{dp} composition, analyzing a large organization's use of \gls{dp} at the global scope (i.e., across all releases) will generally result in a privacy budget that is so high as to be meaningless, even if the budget is adjusted appropriately for the context.
On the other hand, only considering releases in isolation, as in current practice, introduces non-trivial privacy risk due to implicit independence assumptions.
As a result, it may seem that we are fundamentally limited to \emph{either} providing guarantees that provide meaningful privacy parameters (but at a very narrow scope) \emph{or} that consider the cumulative impact at a broad scope (but with excessively high privacy parameters).
While we cannot overcome the fundamental constraints of \gls{dp} for any given scope, we can nevertheless improve our understanding of the privacy risks involved by considering \gls{dp} guarantees at multiple scopes.
For example, Kifer et al. consider a set of ad-hoc scopes in their ex-post case study on the 2020 Census redistricting data release in an attempt to provide more nuanced guarantees on specific (categories of) privacy-sensitive data~\cite{Kifer2022-censussemantics}.
In \oursystem, we systematize this approach and provide the ability to consider a large number of scopes simultaneously,
as well as the ability to define them in a structured and automated way.

\fakeparagraph{Scope-aware Rules}
In \oursystem, we can represent scopes via (intermediate) \glspl{rule} that select only the mechanisms belonging to that scope, using predicates over the labels of the mechanisms.
This can support arbitrary custom scopes an organization requires, e.g., scopes for specific analysts or purposes.
Beyond such organization-specific scopes, scopes in \oursystem can be used to track guarantees on specific (categories of) sensitive data as in Kifer et al.~\cite{Kifer2022-censussemantics}.
Towards this, we require each mechanism to be labeled with the data attributes it accesses.
This information should already be available, as it is required to perform a meaningful \gls{dp} analysis.
At the most basic level, such scopes can establish per-attribute privacy budgets.
Let $A = \{a_1, a_2, \ldots \}$ denote the attributes of the database.
For each attribute $a_k$, we can establish an (intermediate) \gls{rule} $(\phi_k, U_k, B_k)$, where the predicate $\phi_k(m)$ selects only mechanisms that use the attribute $a_k$.
We require that the composition of the selected mechanism remains below the budget $B_k$ (for privacy unit $U_k$).
This can be seen as a form of attribute \gls{dp}~\cite{Kifer2011-dpnolunch}, where the privacy unit between neighboring databases is re-defined to be any attribute.
While attribute \gls{dp} is usually used with bounded-\gls{dp},
modeling such \glspl{rule} as scope (instead of as unit) allows for such attribute constraints in both bounded- and unbounded-\gls{dp}.\footnote{\cfref{sec:bg} for the difference between bounded- and unbounded~\gls{dp}.}

\fakeparagraph{Per-Attribute Policies}
While it would be possible to define individual per-attribute budgets, an organization will typically have a significant number of attributes, and it would be neither scalable nor meaningful to assign each attribute an individual privacy budget.
Alternatively, existing work on attribute \gls{dp} typically assumes the same budget for every attribute~\cite{Kifer2011-dpnolunch, Ghazi2022-dpcontrol}, which may be too limiting.
In \oursystem, we instead introduce privacy risk levels that classify attributes in the data schema based on their sensitivity.
Each risk level $r$ (e.g., \textit{low}, \textit{medium}, and \textit{high}) is associated with a budget $B_{r}$.
By default, attributes receive their budget based on risk levels.
However, \glspl{admin} can alternatively also provide custom budgets for specific attributes.
A per-attribute \gls{policy} (c.f. \Cref{fig:policy:perattr}) defines a per-attribute budget of $B_{r}$ for a privacy risk level $r$ and refers to the set of attributes $A_{r} = \{a_1, a_2, \ldots\}$ with this risk level.
For each per-attribute \gls{policy}, \oursystem automatically derives the per-attribute intermediate \gls{rule} set $\hat{R}_{A}$, which contains a \gls{rule} per attribute.
The subset of \glspl{rule} $\hat{R}_{A_r} \subseteq \hat{R}_{A}$, corresponding to attributes at the same risk level $r$, all share the same budget:
\begin{equation*}
    \hat{R}_{A_r} = \{(\phi_{1}, U_1, B_{r}), (\phi_2, U_2, B_{r}),  \ldots,  (\phi_{|A_r|}, U_{|A_r|}, B_{r})\}
\end{equation*}

\fakeparagraph{Attribute Category Policies}
Per-attribute scopes by themselves are limited in their expressiveness, as they do not model the dependencies (or lack thereof) between different data attributes.
In \oursystem, we use \emph{categories} to make these traditionally implicit independence assumptions explicit.
Specific \glspl{category} might include biometrics, financial, health, personal, demographics, activity, location, etc. but will naturally be organization-specific.
Similar to the attribute privacy risk levels, our system offers the ability to use \gls{category} risk levels to remove the requirement of defining a budget for every \gls{category} individually.
An attribute might naturally belong to multiple categories,
but even then, cleanly assigning attributes into \glspl{category} may often be difficult due to edge cases, where an attribute is \emph{related} to a \gls{category} but not strictly a member.
To address these edge cases, \oursystem supports multiple \glspl{level}, for example: \textit{member}, \textit{strong connection}, and \textit{weak connection}.
Each attribute can be either a direct member of a \gls{category}, have a \textit{strong connection} or \textit{weak connection}, or be assumed to be independent.

\oursystem translates each category into separate cumulative (intermediate\footnote{These rules are further extended by context policies.}) \glspl{rule}: \emph{(i)} a \gls{rule} for members, \emph{(ii)} a \gls{rule} for members and attributes with a strong connection, and \emph{(iii)} a \gls{rule} for attributes with any connection.
The \gls{category} \gls{rule} predicate $\phi(m)$ for the \gls{level} \textit{strong connection}, for example, selects any mechanism relying on an attribute that is a member of the \gls{category}, or has a strong connection to the \gls{category}.
While it is possible to configure a budget for each \gls{category} risk level, and for each \gls{level}, \oursystem applies the budget extension functions introduced for context (\cfref{sec:policy:ctx}), also for extending the budget from the member \gls{level} to the strong-connection and weak-connection \glspl{level}.
This allows an \gls{admin} to explicitly encode their assumptions about when cumulative privacy loss tracking is necessary, where higher privacy loss can be tolerated given weaker dependencies, and when independence assumptions are justified as attributes are unrelated.

A \gls{category} \gls{policy} (c.f. \Cref{fig:policy:percat}) defines per-category budgets $B_r$ based on the category risk level $r$. %
Let $C_r = \{c_1, c_2, \ldots \}$ denote the \glspl{category} with \gls{category} risk level $r$, and let $L = \{l_1, l_2, l_3\}$ denote the \glspl{level}.
A \gls{rule} for the level $l_j$ of category $c_i$ is a tuple $(\phi_{ij}, U_i, f_j)$ consisting of a predicate over the mechanisms $\phi_{ij}: \mathcal{M} \rightarrow \{0,1\}$, and a function $f_j: B \rightarrow B$ on privacy budgets for unit $U_i$.
\oursystem derives the (intermediate) \gls{rule} set $\hat{R}_{C}$,  which contains $|L|$ \glspl{rule} per category.
The subset of \glspl{rule} $\hat{R}_{C_r} \subseteq \hat{R}_{C}$ corresponding to \glspl{category} of the same risk level $r$, all share the same budget $f_j(B_r)$ per membership \gls{level} $j$:
\begin{equation*}
R_{C_r} = \left\{
\begin{aligned}
    &  && (\phi_{11}, B_r), \; (\phi_{12}, f_2(B_r)), \; (\phi_{13}, f_3(B_r)), & \\
    & && (\phi_{21}, B_r), \; (\phi_{22}, f_2(B_r)), \; (\phi_{23}, f_3(B_r)), & \\
    & && \qquad\qquad\qquad\quad \vdots &
\end{aligned}
\right\}~\footnote{With a slight abuse of notation, we abbreviate the identity function for the member level as $B_r = f_{id}(B_r)$, and omit the unit $U_i$ for conciseness.}
\end{equation*}
Note that per-attribute and per-category \glspl{policy} can be seamlessly combined with context \glspl{policy}, as they simply define new intermediate rules, and the extension \glspl{policy} introduced in \Cref{sec:policy:ctx} apply to the union of all intermediate \glspl{rule} defined by the base \glspl{policy}.

\subsection{Multiple Privacy Units}
\label{sec:policy:unit}

So far, we have mostly omitted the privacy unit from our considerations and treated scope- and context \glspl{policy} and the \glspl{rule} they generate as if they were based on a single, universally used privacy unit.
In real-world applications, however, data releases often use different privacy units~\cite{Desfontaines2022-dppractice}.
In \oursystem, we can quantify and control the privacy guarantee associated with multiple privacy units simultaneously,
which \gls{dp} experts consider essential for future use of \gls{dp}~\cite{Cummings2023-dpfrontier}.
Specifically, \oursystem supports arbitrary privacy units as long as each mechanism $m \in \mathcal{M}$ can calculate its privacy cost across all privacy units used in the system.
Typically, user-level privacy units are relaxed across three dimensions (\cfref{sec:bg}):
\emph{(i)} attribute-based privacy units, %
\emph{(ii)} release-specific privacy units, %
and \emph{(iii)} time-based privacy units. %
In \oursystem, we can express the semantics of attribute-based privacy units using scope \glspl{policy}.\footnote{The advantage of scope \glspl{policy} is that it naturally allows per-attribute (group) restrictions for bounded and unbounded \gls{dp} (\cfref{sec:bg}).}
Release-specific privacy units inherently lack meaning across releases and can, therefore, only be used to report guarantees of individual releases.
In the following, we focus specifically on time-based privacy units, which are frequently used in real-world \gls{dp} deployments.
Time-based privacy units are based on the assumption that it is difficult to link users across time boundaries (e.g., that a user's data distribution changes significantly over time).
They require that each record has a user ID and a timestamp convertible to the desired time unit, which (in combination) serves as the record's privacy ID.
While timestamps are inherent in data streams, static data (e.g., country of residence) generally lacks meaningful timestamps.
However, it is common practice to treat static data as if it were part of the data stream for each time unit. %
For example, for page view data, each record contains the geo-location and the browser, even though the location and browser do not change with every visit of a page.
However, the independence assumption of the privacy unit will not hold for such static data, and so the \gls{dp} guarantee relative to this privacy unit does not provide additional benefits.

\fakeparagraph{Converting between Privacy Units}
In practice, it might be hard to judge for which (if any) time-based privacy unit the independence assumptions hold,
e.g., at which time scales the data distribution of all users incurs sufficient changes between two units to be considered independent.
However, \gls{dp} guarantees for time-based privacy units can be converted to guarantees for different time scales, allowing us to understand how these guarantees change depending on which assumptions we are willing to make.
By group privacy (\cfref{sec:bg}), a privacy parameter for a privacy unit $u_{src}$ implies a privacy parameter for a privacy unit $u_{dst}$ if $u_{dst}$ can be seen as a combination of a group of $k$ instances of $u_{src}$.
For example, converting from user-day to user-month requires applying group privacy with a group size of 31.
Note that we can technically also convert to smaller time scales via group privacy, e.g., a user-month guarantee implies the same (group size one) privacy parameter for the user-day privacy unit.
Converting from user-month to user-week, however, would require a group size of two, as a week might span two months.
Note that the conversion using group privacy provides only an upper bound on the privacy loss, while the loss may be much lower concretely.
For example, Wikimedia currently releases daily page view statistics with $\rho_d = 0.015$, equivalent to $(\epsilon_d=0.8, \delta=10^{-6})$, using a UserDay privacy unit.~\footnote{We convert the \gls{zcdp} guarantees to \gls{adp} using a (non-tight) conversion~\cite{Desfontaines2024-dpconverter} for a fixed $\delta=10^{-6}$ to aid intuitive understanding.}
Under group privacy, this daily guarantee implies a UserMonth privacy loss of $\rho_m = 31^2 \cdot \rho_d = 14.415$, or $(\epsilon_m=41.94, \delta=10^{-6})$, which may be excessive and overly pessimistic,
as group privacy must assume users view Wikimedia to the same extent every single day.
However, user activity often fluctuates, with high activity on some days and lower activity on others~\cite{Wikipedia2024-siteviews}.
As a result, one would expect to be able to achieve much better guarantees than what can be achieved purely via group privacy.

\fakeparagraph{Quantifying Privacy for Multiple Privacy Units}
Improving upon group-privacy-based bounds requires adapting mechanisms to consider multiple privacy units simultaneously.
For example, if, in addition to the contribution bound of 10 visits per user and day considered in the Wikimedia release~\cite{Adeleye2023-dpwikipedia}, we also bound the monthly visits to 70 we can provide a UserMonth guarantee of $\rho_M=0.735$, e.g., $(\epsilon_{m}=6.72, \delta=10^{-6})$, significantly improving on the group-privacy-based bound without meaningfully impacting utility.
In Appendix \ref{sec:appendix:units}, we discuss how the programming framework~\cite{Hay2020-opendpprog} underlying the state-of-the-art \gls{dp} libraries~\cite{Berghel2022-tumult, OpenDP2020-whitepaper} can be generalized to support multiple privacy units to achieve these improved bounds.

\oursystem leverages this approach to enable policies to constrain budgets for small and large privacy units simultaneously,
which is essential for fine-grained privacy risk management.
For example, only considering small privacy units does not provide a useful understanding of how guarantees are affected should the underlying assumptions fall short.
On the other hand, only considering larger privacy units is also insufficient:
releases might concentrate the comparatively high privacy-loss parameters necessary to express guarantees at this level into a small time scale, violating privacy expectations.
By considering both small and large guarantees simultaneously, \oursystem provides a significantly more nuanced approach to privacy units.

%% file: sections/enforcement.tex
\subsecspacingtop
\section{Policy Enforcement}
\subsecspacingbot
\label{sec:enforce}

In this section, we consider the enforcement of a \gls{policy} set expressed in the \gls{policy} language for organization-wide privacy risk management introduced in \Cref{sec:policy}.
In \oursystem, a \gls{policy} set defines a (potentially very large) number of scope- and context-aware \glspl{rule},
and the system needs to check that no \gls{rule} is violated for any series of releases.
We first present an optimization that reduces the size of this \gls{rule} set while maintaining the privacy semantics,
generically reducing the problem size for any enforcement algorithm.
Afterward, we describe how \oursystem efficiently enforces complex \glspl{policy} even while providing a fine-grained privacy analysis.
Finally, we discuss how \oursystem can be employed to manage the privacy risk of complex one-off releases that include multiple mechanisms (e.g., queries), or be integrated with existing privacy management systems that allocate limited budgets among \gls{dp} applications.

\subsection{Rule Pruning}
\label{sec:enf:prune}
The challenge for \gls{policy} enforcement is that it naturally scales with the size of the \gls{rule} set.
Even without considering per-release \glspl{policy}, the set of \glspl{rule} $R$ derived from the \glspl{policy} can become very large.
\begin{equation}
    \label{eq:ruleset}
    R = \{(\phi_1, U_1, B_1), \; (\phi_2, U_2, B_2), \; \ldots \}
\end{equation}
Specifically, if the \gls{policy} set tracks: \emph{(i)} $|U_{nit}|$ privacy units (\cfref{sec:policy:unit}), \emph{(ii)} separate scopes for $|A_{ttr}|$ attributes, $|C_{at}|$ categories each with $|L_{vl}|$ membership levels, and potentially a global scope~(\cfref{sec:policy:scope}), and \emph{(iii)}  in total $|E_{xt}| = |E_1| * |E_2| * \ldots$ extensions (\cfref{sec:policy:ctx}) corresponding to the product of the number of settings from each extension, the \oursystem \gls{rule} set has size:
\begin{equation*}
|R| = |U_{nit}| * (1 + |A_{ttr}| + |C_{at}| * |L_{vl}|) * |E_{xt}|
\end{equation*}

Since \oursystem needs to check each \gls{rule} for each release, the large number of \glspl{rule} could become a scalability problem when deploying \oursystem.
However, there is a significant optimization potential to reduce the size of the \gls{rule} set by considering that the privacy guarantee of a \gls{rule} might already be implied by another \gls{rule}.
We identify such \glspl{rule} and prune them from the \gls{rule} set, as they cannot become a deciding factor in a \gls{policy} decision.
For example, it is easy to see if the \gls{rule} set contains two \glspl{rule} $(\phi_1, U_1, B_1)$ and $(\phi_2, U_2, B_2)$ that have the same privacy unit and equivalent predicates, it is sufficient to track only the more restrictive \gls{rule} with the smaller budget and prune the other.
More generally, we can define a partial order on \glspl{rule} that allows us to formally describe which rules can be pruned safely.

\fakeparagraph{Formal Definition}
The predicate set $\Phi$ forms a partially ordered set $(\Phi; \sqsubseteq)$, with the partial order relation $\sqsubseteq$, where $\forall \phi_1 \in \Phi, \; \forall  \phi_2 \in \Phi$:
\begin{equation*}
\phi_1 \sqsubseteq \phi_2\iff\{m \in \mathcal{M} | \;\phi_1(m) \} \subseteq \{m \in \mathcal{M} | \; \phi_2(m) \}
\end{equation*}
Similarly, a privacy unit set $U$ also forms a partially ordered set $(U; \leq)$, $u_1 \leq u_2$ if privacy unit $u_2$ covers $u_1$.
For example, for the set $U = \{UserWeek, UserMonth, User \}$, the partial order would be $UserMonth \leq User$ and $UserWeek \leq User$, but there is no order between $UserWeek$ and $UserMonth$, as a week can span two months.
By combining these two partial orders, the set of rules $R$ forms a partially ordered set $(R; \preceq)$, with the partial order relation $\preceq$, where $\forall r_1 = (\phi_1, U_1, B_1) \in R, \; \forall  r_2=(\phi_2, U_2, B_2) \in R$:
\begin{equation}
\label{eq:rulerelation}
r_1 \preceq r_2\iff \phi_1 \sqsubseteq \phi_2 \; \land \; U_1 \leq U_2
\end{equation}
If the base \glspl{policy} define a global scope, and each extension \gls{policy} contains an extension with $\phi_* \mapsto 1$, and the privacy units have a maximal element, e.g., User, then the  $(R; \preceq)$ is a bounded join-semi-lattice, where $r_* = (\phi_*, u_*, B_*)$ is the greatest element.

\begin{definition}
\label{def:nonconstrain}
In a \gls{rule} set $R$, a \gls{rule} $r$ is called non-constraining if there exists no composition of mechanisms $M = (m_1, m_2, \ldots, m_N)$ that satisfies all \glspl{rule} in $R \setminus \{r\}$ but violates \gls{rule} $r$.
\end{definition}
\noindent
Due to the structure of the \gls{rule} set, we can determine whether a \gls{rule} is non-constraining purely by considering the partial order relation and budgets.
\begin{theorem}%
\label{thr:cover}
In a \gls{rule} poset $(R, \preceq)$, a rule $r_i = (\phi_i, u_i, B_i) \in R$ is non-constraining if
$\exists r_j \in R \setminus \{r_i\}$ such that $r_i \preceq r_j$ and $B_i \geq B_j$.
\end{theorem}

\begin{proof}
  See Appendix~\ref{sec:appendix:coverproof}.
\phantom\qedhere
\end{proof}

\begin{figure}[t]
\centering
\begin{tikzpicture}

    \tikzstyle{graynode}=[gray]
    \tikzstyle{boldedge}=[line width=2pt]

    \node (1) at (0,3) {$(\phi_1, u_1,\epsilon=7)$};

    \node[graynode] (2) at (-2.5,1.5) {$(\phi_2, u_2, \epsilon=7)$};
    \node[graynode] (3) at (0,1.5) {$(\phi_3, u_3, \epsilon=7)$};
    \node (4) at (2.5,1.5) {$(\phi_4, u_4, \epsilon=5)$};

    \node[graynode] (5) at (-2.5,0) {$(\phi_5, u_5, \epsilon=7)$};
    \node (6) at (0,0) {$(\phi_6, u_6, \epsilon=3)$};
    \node[graynode] (7) at (2.5,0) {$(\phi_7, u_7, \epsilon=5)$};

    \draw[boldedge] (1) -- (2);
    \draw[boldedge] (1) -- (3);
    \draw (1) -- (4);

    \draw[boldedge] (2) -- (5);
    \draw (2) -- (6);
    \draw (3) -- (6);
    \draw[boldedge] (4) -- (7);
\end{tikzpicture}
\caption[Hasse diagram of a \gls{rule} poset $(R, \preceq)$.]{Hasse diagram of a \gls{rule} poset $(R, \preceq)$: non-constraining \glspl{rule} are shown in grey, with the path to the more general \gls{rule} under a stricter budget highlighted in bold.
}
\label{fig:hasse}
\end{figure}
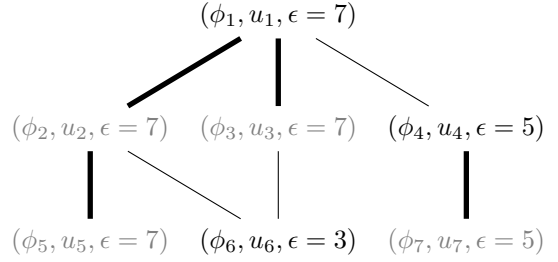

In \Cref{fig:hasse}, we show an example of a \gls{rule} poset, distinguishing between constraining and non-constraining \glspl{rule}.
\Cref{thr:cover} lies the groundwork for an algorithm that sequentially identifies a \gls{rule} that is non-constraining and thus can be pruned, repeating this until there are no non-constraining \glspl{rule} left.
However, implementing this algorithm efficiently hinges on  evaluating $r_1 \preceq r_2$ efficiently, which requires computing $\phi_1 \sqsubseteq \phi_2$ and $U_1 \leq U_2$ (Eq. \ref{eq:rulerelation}).
While the partial order for the privacy units can be explicitly encoded by the \gls{admin} as the set is generally small,
determining the partial order between the predicates is more involved.
Determining whether $\phi_1 \sqsubseteq \phi_2$ on arbitrary predicates is a hard problem as it requires showing that $\phi_1 \land \neg \phi_2$ is not satisfiable.

\fakeparagraph{\Gls{rule} Decomposition}
In \oursystem, we can exploit the fact that the vast majority of \glspl{rule} are generated from extension \glspl{policy} on base \glspl{policy} (c.f. \Cref{eq:rulepreddecomp}) to significantly simplify this computation.
This allows decomposing the problem of determining $\phi_{1} \sqsubseteq \phi_{2}$ for such \glspl{rule}.
Specifically, the poset $(\Phi; \sqsubseteq)$ over these \glspl{rule} can be decomposed into a poset $(\Phi^{\ell + 1}; \leq)$ defined on $\Phi_{B} \times \Phi_{1} \times  \ldots \times  \Phi_{\ell}$, where:
\begin{equation*}
    \begin{aligned}
    & (\phi_{B}^{(1)}, \phi_{1}^{(1)}, \ldots, \phi_{\ell}^{(1)}) \;\leq \; (\phi_{B}^{(2)}, \phi_{1}^{(2)}, \ldots, \phi_{\ell}^{(2)})
    \\
    & \iff
    \phi_{B}^{(1)} \sqsubseteq \phi_{B}^{(2)} \land \phi_{1}^{(1)} \sqsubseteq \phi_{1}^{(2)} \land \;\ldots \; \land \phi_{\ell}^{(1)} \sqsubseteq \phi_{\ell}^{(2)}
    \end{aligned}
\end{equation*}

\noindent
This decomposition allows us to significantly reduce the number of predicates we need to compare by applying basic memoization.
Let $\hat{R}$ denote the set of intermediate \glspl{rule} generated by the base \glspl{policy} (\cfref{sec:policy:scope}), and let $E_i$ denote the set of extensions from the $i$-th extension \gls{policy} (\cfref{sec:policy:ctx}).
Instead of computing the predicate relation $\phi_1 \sqsubseteq \phi_2$ for each pair of \glspl{rule}, which would require $\mathcal{O}\left(\left(|\hat{R}| * |E_1| * \ldots * |E_\ell|\right)^2\right)$ comparisons, this decomposition reduces the number of comparisons to $\mathcal{O}\left(|\hat{R}|^2 + |E_1|^2 + \ldots + |E_\ell|^2\right)$.
Beyond this asymptotic advantage, decomposing the predicates in this way also makes it straightforward to compute most of the individual partial order relations.
Each extension \gls{policy} typically defines only a small set of extensions (e.g., a ``standard'' and a ``black-box'' \gls{ml} setting), and \glspl{admin} can easily
annotate each extension to define the partial order.
\footnote{An annotation is an integer tuple such that one extension is smaller or equal to another if each integer in its tuple is smaller than or equal to the corresponding integer in the other's tuple.}
In contrast to extension \glspl{policy}, base \glspl{policy} generate a large intermediate \gls{rule} set defining different scopes.
To compute $\phi_{B}^{(1)} \sqsubseteq \phi_{B}^{(2)}$, requires a partial order on the predicates of these base \glspl{rule}.
However, both per-attribute \glspl{policy} and attribute category \glspl{policy}, which are responsible for the large majority of predicates, generate predicates where each one simply specifies a set of attributes to which its scope applies.
Thus, we can use the subset relation on these attribute sets to establish the partial order.
For additional custom scoping \glspl{policy}, we again rely on the \gls{admin} to annotate the order directly (potentially using the help of a SAT solver if necessary).

\begin{algorithm}[t]
\begin{algorithmic}[1]
\Function{RuleSetPruning}{rules}
    \State $\top$ $\leftarrow$ \Call{greatest}{rules}
    \State \Call{DeactivateCoveredRules}{$\top$, $\top$.budget}
\EndFunction
\State
\Function{DeactivateCoveredRules}{rule, budget}
    \For{crule \textbf{in} \Call{LowerCover}{rule}}
        \If{budget $\leq$ crule.budget}
            \State \Call{DeactivateRule}{crule}
            \State cbudget  $\leftarrow$ budget \Comment{implied budget}
        \Else
            \State cbudget $\leftarrow$ crule.budget
        \EndIf
        \State \Call{DeactivateCoveredRules}{crule, cbudget}
    \EndFor
\EndFunction
\end{algorithmic}
\caption{The \oursystem \gls{rule} set pruning algorithm.}
\label{algo:prune}
\end{algorithm}

\fakeparagraph{Pruning Algorithm}
In \Cref{algo:prune}, we introduce the \oursystem \gls{rule} set pruning algorithm, designed to identify and deactivate \glspl{rule} that can be safely removed according to \Cref{thr:cover}.
The algorithm takes as input a partially ordered set of \glspl{rule}  $(R; \preceq)$ and starts at the greatest element.\footnote{Note that by including the global scope $\phi_* \mapsto 1$ as a custom base \gls{policy}, and tracking a greatest privacy unit (e.g., user), the poset forms a bounded join semi-lattice, which guarantees a greatest element. Otherwise, a special $\infty$ element can be added to ensure a maximal element exists.}
The algorithm proceeds recursively to smaller \glspl{rule}, pruning those that do not constrain the composition further while also updating (implied) budgets as it traverses.
For \gls{rule} $r$, we determine its lower cover, i.e.,
find the \glspl{rule} $r_c \in R \setminus \{r\}$, for which no intermediate \gls{rule} $r_i \in R \setminus \{r, r_c\}$ exists such that $r_c \leq r_i \leq r$.
For each of the \glspl{rule} in the lower cover of $r$, we compare whether the \gls{rule} has a stricter budget than the budget of the \gls{rule} $r$, and otherwise prune the \gls{rule}.
When a \gls{rule} is pruned, it means that there is a greater \gls{rule} with a stricter budget that also implies a stricter budget for this \gls{rule} and all its successors; hence, we propagate this implied budget further.
Specifically, for every \gls{rule} $r_c$ in the lower cover, we call the recursive function with the (implied) budget.
The recursion stops when it reaches minimal \glspl{rule}, which have no lower cover.
In this algorithm, a \gls{rule} is only pruned if a more general \gls{rule} has a stricter budget, and hence, according to \Cref{thr:cover}, can be safely removed.
By propagating the (implied) budget downward, also \glspl{rule} are pruned where the covering \gls{rule} is not a direct predecessor.

\newglossaryentry{monitor}{name={policy decision point}, plural={policy decision points}, description={}}

\subsection{Policy Decision Point}
\label{sec:enf:monitor}

In this section, we discuss the design of a \gls{monitor} for enforcing a \gls{policy} set.
The \gls{monitor} determines whether to allow or reject release requests, %
where each release request contains one or more labeled \gls{dp} mechanisms and the multiset of labels for each mechanism defines its context and associated scopes (\cfref{sec:policy}).
The \gls{monitor} proceeds in two stages, first checking each mechanism's compliance with the configured per-release \glspl{policy}.
If any per-release \gls{rule} is violated, the entire release request is denied.
The \gls{monitor}'s second stage checks the \glspl{policy} that take into account the cumulative impact of releases.
In \gls{dp} terminology, this operates akin to a privacy filter, which, given a sequence of adaptively chosen mechanisms, ensures that a pre-specified budget for privacy parameters is not exceeded~\cite{Rogers2016-dpodometer, Lecuyer2021-dpodometers, Haney2023-dpfadaptconc}.
For the remainder of the section, we focus on the second stage of the \gls{monitor}, as enforcing per-release \glspl{policy} is straightforward.

\fakeparagraph{\Glspl{rule} Checking Engine}
While the \gls{rule} set optimization discussed in \Cref{sec:enf:prune} can significantly prune the \gls{rule} set, a considerable number of \glspl{rule} may remain to be checked for every release request.
Determining whether a given \gls{rule} applies to a mechanism requires evaluating the \gls{rule}’s predicate on the mechanism’s labels.
Instead of performing a linear scan through all \glspl{rule}, we can exploit the fact that the \gls{rule} set forms a partially ordered set $(R; \preceq)$ (\cfref{sec:enf:prune}).
The algorithm starts at the maximal \glspl{rule} in the poset
and proceeds to smaller rules, evaluating the predicate when visiting a \gls{rule} for the first time.
If the predicate is true, it checks the \gls{rule} for a potential violation and continues to the \glspl{rule} in its lower cover.
If the predicate is false, we can skip not only that \gls{rule} but also all smaller \glspl{rule} in the poset.
This algorithm is sufficient to integrate \oursystem with simple budget tracking systems where checking release requests only requires comparing the privacy cost of the mechanism with the remaining budget.
For example, this approach is used in Tumult Analytics~\cite{Berghel2022-tumult}, and we could easily configure the session with a \oursystem \gls{policy} set instead of a global privacy budget.
However, this simple budget-tracking approach does not allow for automatically applying parallel composition \emph{across} queries.

Using a more fine-grained accounting approach based on block composition~\cite{Lecuyer2019-sage} can leverage cross-query data access patterns to achieve a tighter privacy analysis.
For example, if two queries each access only individuals from two different countries, then block composition allows the costs to be tracked in two separate blocks, one for each of the two countries.
More generally, a set of \glspl{pa} can be defined that partition users (or other privacy units) into blocks.
If we track privacy loss for each possible block\footnote{Note that we need to consider the full domain of \glspl{pa}, as data access patterns are generally not known in advance.} separately, we can fully leverage the potential of parallel composition this introduces.
However, this straightforward approach scales poorly, especially so when considering that, in \oursystem, we would need to track a full set of blocks for each \gls{rule}.
Instead, we build upon an optimized tracking approach proposed by K\"uchler et al.~\cite{Kuchler2024-cohere}, which allows dynamically adjusting the granularity of the block tracking based on the concrete set of mechanisms and their data access patterns.
We extend this idea to include \glspl{rule} as another dimension alongside the domain of \glspl{pa}, applying the proposed segmentation algorithm on the combined space.
Rather than tracking a large set of individual blocks for each \gls{rule}, we can track a small number of segments that can span multiple blocks across \glspl{rule}, significantly reducing the total number of privacy filters we need to consider.
In the worst-case, the total number of segments to track remains unchanged, however, in practice, we anticipate notable reductions by integrating \glspl{rule} and \glspl{pa}, as requests with similar data-access patterns can be expected to match similar \glspl{rule}.

While each \gls{rule} can, in principle, define its own unique set of \glspl{pa}, this level of flexibility is often unnecessary as typical \glspl{pa} (e.g., country of origin), are not \gls{rule}-specific.
However, for \glspl{rule} with a time-based privacy unit (\cfref{sec:policy:unit}), integrating the time steps as a \gls{pa} can be beneficial.
By considering time steps as a \gls{pa}, we can also leverage parallel composition on time steps across queries.
For composition efficiency, ideally, we would track each time step individually; however, this poses a challenge for tracking because time steps are an infinite sequence.
To balance tracking complexity with composition efficiency, we track specific (e.g., recent) time steps in a more granular level and collapse more distant time steps into broader intervals.
For example, let $t_1, t_2, \ldots, t_\infty$ denote the sequence of time steps in a time-based privacy unit (e.g., months).
We might represent this as $[t_1, t_{k-1}], t_{k}, t_{k+1}, \ldots, t_p, [t_{p+1}, t_\infty]$.~\footnote{The inner time steps can also be collapsed into intervals if desired.}
This exploits the natural tendency of requests to consider ``current'' time steps more selectively, while requests for ``historical'' time steps are more likely to span across larger intervals already, in which case there is no loss of tracking precision.

\subsection{System Integration}
\oursystem can be employed to manage the privacy risk of complex one-off releases that include multiple mechanisms (e.g., queries), or be integrated with existing privacy management systems to support ongoing data releases over time.
Several systems already provide authorization workflows for \gls{dp} releases~\cite{Aymon2024-lomas, OpenMinded2024-pysyft}, requiring a data owner or privacy officer to approve or deny analysts' release requests.
\oursystem enhances these workflows by enabling the control and monitoring of privacy risks across multiple releases. %
Note that privacy risk management inherently depends on high-quality data and request annotations.
In \oursystem we assume mechanisms are labeled with the correct privacy loss, scope, and context, however, hardening the approval process is an interesting direction for future work.
For example, combinations of view-based access control and data management systems with support for fine-grained logical partitioning of the database~\cite{Albab2023-k9db} may be promising in addressing this gap.
Recent work has also investigated advanced automated privacy management systems that treat \gls{dp} budgets as resources and frame the decision process as an allocation problem~\cite{Kuchler2024-cohere, Luo2021-privacysched}.
This includes solutions that can prevent budget depletion, even for finite user-level \gls{dp} budgets, through techniques such as user rotation.
\oursystem integrates seamlessly with these approaches, as the constraints imposed by our system can be expressed as a variant of the multidimensional knapsack formulation which also underlie the block-composition-based allocation problem that these solutions operate on.

\fakeparagraph{Policy Management \& Governance}
The effectiveness of \oursystem depends to a large extent on a well-defined policy set.
In general, navigating the privacy-utility tradeoff in \gls{dp} is non-trivial and has been studied extensively~\cite{Miklau2022-dpcasesalary, Avent2019-dpareto}.
However, this challenge already occurs when considering single releases and is, for the most part, orthogonal to our work.
Nevertheless, managing \gls{dp} at an organizational scale does introduce new complexities.
Similar challenges occur in all large-scale policy-based systems, including well-established domains such as access control and firewall configuration.
Just as a misconfigured access control rule can cause data leakages, a misconfigured \gls{dp} policy can result in unintended privacy loss.
These risks can be mitigated through widely established best practices, e.g., incorporating multiple decision points and robust review processes.
As \gls{dp} matures, we anticipate such practices becoming more standardized, as has been the case in other domains.

In an ideal world, an organization would adopt \oursystem before any releases occur, fully defining all \glspl{policy} including contexts and scopes in advance, all data would be properly tagged, and the data schema would remain fixed permanently.
However, in practice, this is unrealistic as organizations frequently have prior releases, and correctly tagging all data at once is a significant challenge.
Instead, \oursystem supports incremental rollouts, including accounting for prior \gls{dp} releases.
\footnote{
\gls{dp} composition theorems are usually formulated for an ideal world. Recent work has shown that most state-of-the art composition theorems also hold for less idealized models~\cite{Haney2023-dpfadaptconc}, and extending this to even more realistic settings remains an interesting avenue for future work.}
An organization can re-analyze prior data releases within \oursystem similar to how a new release would be analyzed, i.e., identifying the relevant attributes, categories, contexts, and defining the privacy units.
The only difference between past and future releases is that past releases have already consumed privacy budget, so they cannot be rejected and their costs must be accounted for.
As more releases are considered, the initial set of \glspl{policy} can be refined, with most updates proceeding smoothly.
However, some changes may lead to \gls{rule} conflicts, where a \gls{rule}'s privacy budget is exceeded based on the previously accepted set of requests.
In this case, \oursystem identifies the \gls{policy} elements at fault, offering a choice of either discarding the change or adjusting the conflicting \glspl{policy}.

More generally, it is always possible to loosen restrictions (e.g., increase privacy budgets, drop attributes from categories, or remove entire categories)
or to introduce new restrictions on attributes, categories, contexts, etc. that have not been used in prior releases.
Conflicts can only occur when attempting to tighten constraints affecting past releases.
Clearly, one cannot undo privacy losses incurred by past releases, yet \gls{policy} adjustments are possible with some caveats.
For example, it is possible to move an attribute to a different category, though this might require increasing that category's budget.
One can also introduce a new context even for past releases, though this may require adding new labels to past releases.
Beyond this, it is even possible to introduce new privacy units, though privacy costs for past releases can only be computed using group privacy, which may be suboptimal.
A more precise analysis would require additional contribution bounds, which generally cannot be introduced retroactively (\cfref{sec:policy:unit}).

%% file: sections/evaluation.tex
\subsecspacingtop
\section{Evaluation}
\subsecspacingbot
\label{sec:eval}
In this section, we evaluate the performance of \oursystem, highlighting how its three core \gls{policy} features effectively mitigate privacy risks that can emerge without comprehensive, organization-wide privacy risk management.

\newcommand{\sRequests}{\ensuremath{\mathcal{R}}\xspace}
\newcommand{\sRequest}{\ensuremath{r}\xspace}
\newcommand{\sAttributes}{\ensuremath{A}\xspace}
\newcommand{\sNumAttributes}{\ensuremath{n}\xspace}
\newcommand{\sNumCategories}{\ensuremath{m}\xspace}
\newcommand{\sNumRequests}{\ensuremath{l}\xspace}
\newcommand{\sNumRequestAttr}{\ensuremath{k}\xspace}

\newcommand{\sa}{\ensuremath{a}\xspace}
\newcommand{\sA}{\ensuremath{A}\xspace}
\newcommand{\sCategories}{\ensuremath{\mathcal{C}}\xspace}
\newcommand{\sC}{\ensuremath{c}\xspace}
\newcommand{\sAssociations}{\ensuremath{\mathcal{L}}\xspace}
\newcommand{\sAssociation}{\ensuremath{l}\xspace}
\newcommand{\mapping}{\ensuremath{S}\xspace}
\newcommand{\catSample}{\ensuremath{\sCategories^\prime}\xspace}
\newcommand{\bernoulli}{\ensuremath{\text{Bernoulli}}\xspace}

\begin{figure*}
    \centering
    \includegraphics[width=1.0\textwidth]{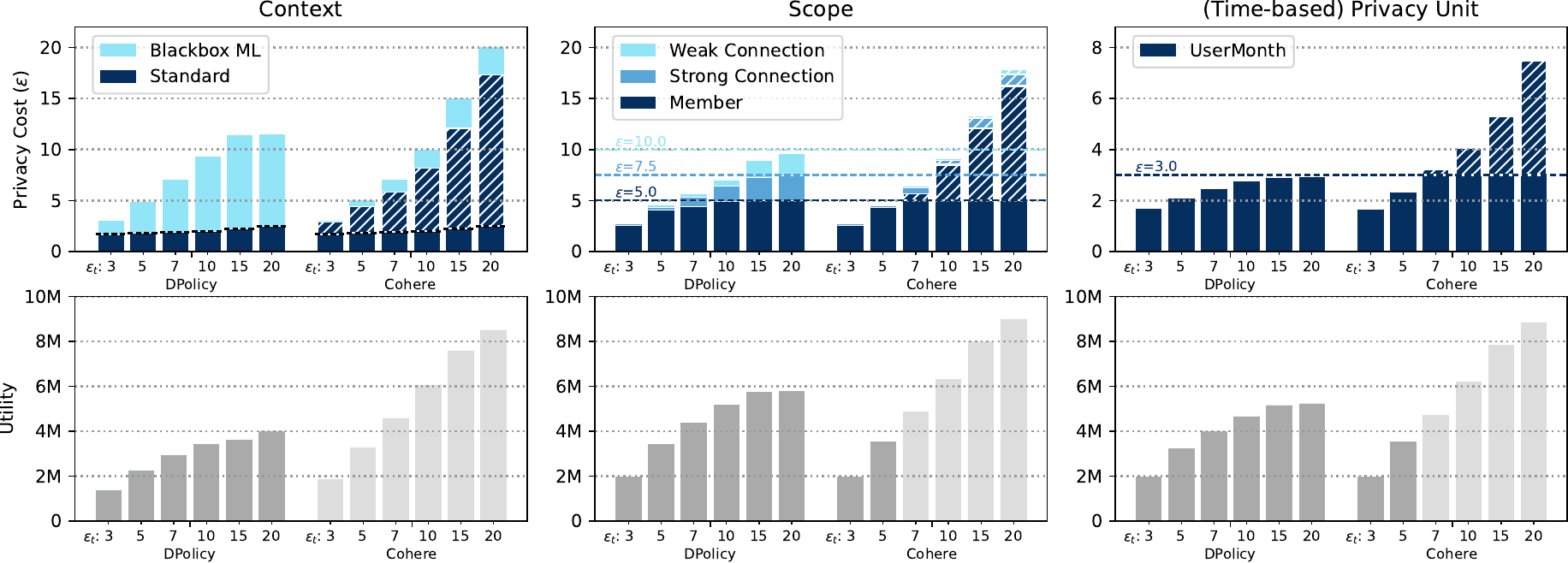}
    \vspace{0em}
    \caption[Privacy cost and utility for \oursystem and Cohere.]{Privacy cost (top) and utility (bottom) for \oursystem and Cohere in our three scenarios.
    For the \gls{trap:context} \gls{trap}, we report the global user-level privacy cost.
    For the  \gls{trap:scope} \gls{trap}, we show the privacy cost of the high-risk category with the largest privacy cost.
    Finally, for the \gls{trap:time} \gls{trap}, we show the user-month privacy cost (for time-based data) for the month with the highest privacy cost.
    We also show the maximum privacy costs acceptable under the \gls{trap}'s policy as a dashed line.
    For Cohere, we indicate privacy costs that violate the \gls{trap}'s \gls{policy} with hatched bars.
    Similarly, we show Cohere's utility in a lighter tone if it was achieved by violating the \gls{trap}'s \gls{policy}.
    }
    \vspace{0.5em}
    \label{fig:dpolicy:eval:ctx}
    \label{fig:dpolicy:eval:scopecat}
    \label{fig:dpolicy:eval:time}
\end{figure*}

\subsection{Evaluation Setup}
We concentrate on a setting where \oursystem is integrated with an advanced privacy management system that treats \gls{dp} budgets as resources requiring careful management, as this represents the most complex deployment scenario.
We compare \oursystem against a state-of-the-art system automated privacy management system, Cohere~\cite{Kuchler2024-cohere}, %
which manages the allocation of privacy resources under a single global user-level budget but does not support more fine-grained \gls{dp} \glspl{policy}.

\fakeparagraph{Implementation}
We implement \oursystem's policy language frontend in Python and use the SageMath library~\cite{sagemath} for policy optimization.
We map the constraints resulting from the policy language to a block-composition-based multidimensional knapsack formulation.
We then solve this formulation using a Rust implementation of the DPK algorithm~\cite{Tholoniat2022-dppacking}.
We also extend Cohere's workload generator with support for sampling the context, attribute, category and time period annotations.
We run \oursystem on a scientific cluster backed by host machines with 128core AMD EPYC CPUs (2.60GHz), with 256 GB memory running Ubuntu 22.04, of which the experiments use 12 virtual cores and 64GB of RAM.
We make all our implementations and benchmarking configurations available as opensource \!\footnote{\url{https://github.com/pps-lab/dpolicy}}.

\fakeparagraph{Evaluation Setting}
While \oursystem is more generally applicable,
for a fair comparison, we evaluate \oursystem in the setting of Cohere~\cite{Kuchler2024-cohere}.
Specifically, we consider unbounded \gls{dp} in the central model, focus on user-level guarantees, and configure privacy budgets in \gls{adp}  with $\delta = 10^{-7}$,  while using \gls{rdp}-based privacy filters\footnote{
We use $\alpha$ orders $\{1.5, 1.75, 2, 2.5, 3, 4, 5, 6, 8, 16, 32, 64, 10^6, 10^{10}\}$.} for composition~\cite{Lecuyer2021-dpodometers}.
We follow Cohere's notion of allocation rounds, simulating 20 weekly rounds in which batches of candidate requests, modeled as a Poisson process with an expected 504 requests per round, compete for privacy resources.
Similarly, we adopt the 12 rounds (i.e., weeks) user and budget unlocking model from Cohere, where, in each round, some users are activated while others are retired, and at least a fixed fraction of the overall privacy budget is guaranteed to be available in each allocation round.

Cohere considers four different workloads modeling different mixes of request.
In our evaluation, we consider their most complex workload (c.f.~\ul{W4:All} with \glspl{pa} in~\cite{Kuchler2024-cohere}) which combines the other three workloads and models the wide variety of different types of \gls{dp} applications that we would expect in a deployment in a large organization.
Specifically, it includes an equal mix of a variety of \gls{dp} mechanisms (Gaussian Mechanism~\cite{Dwork2006-dpgaussian}, Laplace Mechanism~\cite{Dwork2006-originaldp}, Sparse Vector Technique~\cite{Dwork2009-svt}, Randomized Response~\cite{Warner1965-randresp}, DP-SGD~\cite{Abadi2016-mldp}, PATE~\cite{Papernot2018-scalablepate}) and captures variability in privacy requirements across requests by  categorizing requests into three levels of low, medium, and high privacy costs.
In addition, each mechanism category has varying degrees of partitioning, controlling how well requests can leverage parallel composition.
Specifically, Laplace and Gaussian mechanisms represent highly partitioned counting queries.
On the other hand, the sparse vector technique and randomized response mechanisms, which are frequently used in database query tasks, feature less partitioning.
Finally, DP-SGD and PATE represent machine learning tasks with minimal partitioning.
Their workload also assigns each request a utility score, taking into account the request cost and the amount of data accessed.
We refer to Appendix~\ref{apx:eval:config} for a more detailed description of the evaluation setup.

\fakeparagraph{\Glspl{trap}}
We consider three \glspl{trap} that augment Cohere's workload with labels and \glspl{policy}, each highlighting a specific aspect of \oursystem's  privacy risk management. %

\noindent
In \gls{trap:context}, we consider the ``standard'' and ``black-box ML'' settings and define an extension policy (\cfref{sec:policy:ctx}) that relaxes budgets based on the empirical findings of Nasr et al.~\cite{Nasr2021-dpmladv}.
Specifically, we define the budget extensions as
$[ 1.7 \!\mapsto\! 3, 1.8 \!\mapsto\! 5, 1.9 \!\mapsto\! 7, 2.0\!\mapsto\! 10, 2.3\!\mapsto\! 15, 2.5\!\mapsto\! 20]$.
We randomly augment 80\% of the requests belonging to the \gls{ml} mechanisms (i.e., NoisySGD and PATE) with a black-box \gls{ml} label,
while all other requests are labeled as \emph{standard}.

\noindent
In \gls{trap:scope}, we consider ten categories of attributes, in addition to 150 per-attribute scopes.
We set 80\% of the attributes to \gls{risklow}-risk ($\epsilon \leq 20$),  and 10\% each to \gls{riskmedium}-risk ($\epsilon \leq 9$) and \gls{riskhigh}-risk ($\epsilon \leq 3$).
For the categories, we assign half of them budgets of $\epsilon \leq 12$, four a budget of $\epsilon \leq 10$ and one a budget of $\epsilon \leq 5$.
We also define extension functions of $1.5 \cdot b$ for \gls{level:strong}, and $2 \cdot b$ for \gls{level:weak} membership, where $b$ is the \gls{level:member} budget.
In Appendix~\ref{apx:eval:sampling}, we describe how we sample attributes and categories when assigning request labels.

\noindent
In \gls{trap:time}, we consider a user-month privacy unit with a user-month privacy budget of $\epsilon \leq 3$ in addition to the global user-level privacy budget.
Half of the requests select some time-based data from a single month in a window of seven months around the current month.
Requests select the current month with a higher probability ($p = \frac{1}{3}$) than other months ($p = \frac{1}{9}$), modeling a natural tendency for requests to focus on current data.
The other half of the requests selects static data that does not overlap with the time-based requests.

\subsection{Evaluation Results}
We compare \oursystem with Cohere~\cite{Kuchler2024-cohere} on the \glspl{trap} described above, considering a range of total privacy budgets $\epsilon_t \in [3, 5, 7, 10, 15, 20]$ and report the utility and privacy cost of the requests
(c.f. \Cref{fig:dpolicy:eval:ctx}).
The performance of \oursystem (and of Cohere) depends on the complexity of the specific problem instance, with the number of (final) \glspl{rule} and the flexibility afforded by the budgets impacting the allocation complexity.
For the three scenarios we consider, \oursystem finished each weekly  allocation round in less than 15~min. %

\vspace{0.6em}
\fakeparagraph{\gls{trap:context}}
We report the cumulative global user-level privacy cost, differentiating between the cost when considering only ``standard'' requests and the combined cost (including both ``standard'' and ``blackbox ML'' requests).
Naturally, \oursystem enforces the appropriate, tighter, budget on ``standard'' requests (indicated in \Cref{fig:dpolicy:eval:ctx} as a dashed line) while permitting the ``blackbox ML'' requests to use higher budgets.
Cohere, however, lacks this context awareness and allocates ``standard'' requests way beyond this bound, which is clearly inappropriate for this context and would result in significant privacy risk.
Note that instantiating Cohere with the tighter ``standard'' budgets instead would essentially prevent it from allocating ML requests.
On the other hand, \oursystem offers an attractive trade-off between utility and privacy risk management.

\vspace{0.6em}
\fakeparagraph{\gls{trap:scope}}
We report the privacy cost of the high-risk category with the largest privacy cost.
While this is naturally tracked in \oursystem, we compute the equivalent in Cohere by considering the category of allocated requests post-hoc.
In Appendix~\ref{apx:eval:results}, we additionally report results for high-risk attributes.
We indicate the per-category budget ($\epsilon \leq 5$) and its extension to attributes of different membership levels ($\epsilon \leq 7.5$ for strong connections and $\epsilon \leq 10$ for weak connections) with dashed lines in \Cref{fig:dpolicy:eval:scopecat}.
In this \gls{trap}, Cohere trivially satisfies these bounds when the global bound is lower but starts to violate them as the global budget increases,
concentrating undue privacy risks on specific types of data.
\oursystem, on the other hand, performs as expected and uses the increased budget in line with the intended privacy semantics.
This allows our system to achieve considerably higher utility than what Cohere (with $\epsilon \leq 5$) can achieve.

\vspace{0.6em}
\fakeparagraph{\gls{trap:time}}
We report the privacy cost for the user-month privacy unit, which has an intended budget of $\epsilon \leq 3$ (indicated by a dashed line in \Cref{fig:dpolicy:eval:time}).
We omit reporting the (user-level) privacy cost of the non-time-based requests, as no additional constraints are acting on them beyond the global budget.
As a result, both Cohere and \oursystem trivially satisfy the user-level constraints.
For the user-month privacy unit, however, Cohere's allocation eventually violates the intended budget bound as it cannot differentiate between time-based requests for the same vs different months.
Consequently, Cohere cannot prevent the concentration of privacy loss on specific months, increasing privacy risks for those months.
By tracking multiple privacy units and, for the time-based unit, multiple time steps, \oursystem can exploit parallel composition in an additional dimension and achieve more nuanced privacy semantics.

%% file: sections/conclusion.tex
\section{Conclusion}
\vspace{-3pt}
\label{sec:dpolicy:conclusion}
Managing privacy risks across multiple data releases remains a critical challenge for organizations deploying \gls{dp}.
Most deployments today treat each release in isolation, making it difficult to assess cumulative privacy risks across an organization.
While one might expect that defining a single large privacy budget could address this issue, in practice, this approach leads to excessive privacy loss parameters that quickly fail to provide meaningful guarantees.
Moreover, appropriate privacy parameters vary significantly across different contexts, requiring a more flexible and structured approach.
To address these challenges, we introduced \oursystem, a system designed to manage cumulative privacy risks by simultaneously considering multiple \gls{dp} guarantees and making traditionally implicit assumptions about scopes and contexts explicit through a high-level policy language.
\oursystem enables organizations to implement and enforce fine-grained, organization-wide privacy risk management.

%% file: sections/appendix_ieee.tex
\section{Multiple Privacy Units in \gls{dp} Libraries}
\label{sec:appendix:units}
\input{sections/appendix_units}

\section{Proof of \Cref{thr:cover}}
\label{sec:appendix:coverproof}

\input{sections/appendix_proof}

\begin{figure}[t]
    \centering
    \includegraphics[width=0.80\columnwidth]{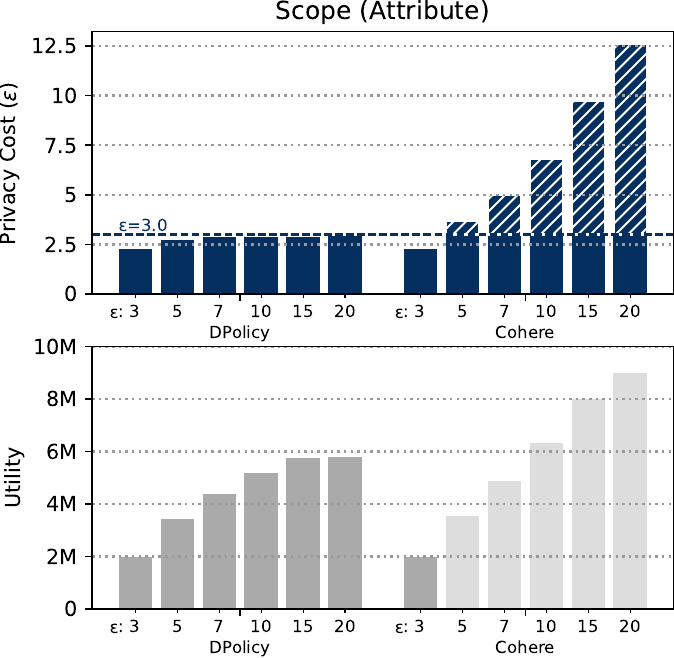}
    \vspace{5pt}

    \vspace{7.5pt}
    \caption{Privacy cost for the largest-cost high-risk attribute (top) and utility (bottom) for \oursystem and Cohere.}

    \label{apx:eval:trap:attribute}
    \vspace{8pt}
\end{figure}

\section{Workload Configuration}
\label{apx:eval:config}

We use the most complex workload \ul{W4:All} from Cohere~\cite{Kuchler2024-cohere} with the same configuration.
\Cref{tab:eval:config} provides a complete overview of all parameters.
Each request has a scalar utility representing its expected organizational value if accepted.
The workload models utility using the Cobb-Douglas production function~\cite{Cobb1928-cobbdouglas}, $Y=A \cdot L^{\beta }K^{\alpha}$, where $L$ is the privacy cost, $K$ is the amount of data, and $A \sim \text{Beta}(0.25, 0.25)$ captures application-level variations.
The parameters are fixed at $\alpha = 1$ and $\beta = 2$.
Requests are assigned one of three privacy cost levels, each with a different $\epsilon$, while $\delta = 10^{-9}$ remains fixed.
Not all requests target all users; instead they select subsets based on \glspl{pa} (e.g., only US users).
The \gls{pa} domain is represented as a vector of size 204,800.
Each request selects a consecutive range of $\lfloor s \cdot 100 \rfloor$ elements, starting at a uniformly random position and wrapping around the edges, with $s \sim \text{Beta}(a, b)$.

\newpage

\section{Attribute and Category Sampling}
\label{apx:eval:sampling}
\input{sections/appendix_scopesampling}

\begin{table}[t!]
    \centering
    \begin{tabular}{llll}
        \toprule
        Mechanism & Privacy Cost $(\epsilon)$ & \gls{pa} \\
        \midrule
        Gaussian Mech.~\cite{Dwork2006-dpgaussian} & $0.05, 0.2, 0.75$ & $Beta(1, 10)$ (high) \\
        Laplace Mech.~\cite{Dwork2006-originaldp} & $0.01, 0.1, 0.25$ & $Beta(1, 10)$ (high) \\
        SVT~\cite{Dwork2009-svt} & $0.01, 0.1, 0.25$ & $Beta(1,0.5)$ (mid) \\
        Rand. Response~\cite{Warner1965-randresp} & $0.01, 0.1, 0.25$ & $Beta(1,0.5)$ (mid) \\
        DP-SGD~\cite{Abadi2016-mldp} & $0.05, 0.2, 0.75$ & $Beta(2, 2)$ (low) \\
        PATE~\cite{Papernot2018-scalablepate} & $0.05, 0.2, 0.75$ & $Beta(2, 2)$ (low) \\
        \bottomrule
    \end{tabular}
    \caption{The workload consists of an equal mix of request categories, each defined by these parameters.}
    \label{tab:eval:config}
\end{table}

\section{Additional Evaluation Results}
\label{apx:eval:results}
We report the privacy cost of the high-risk attribute with the largest privacy cost for \Gls{trap} \gls{trap:scope} in~\Cref{apx:eval:trap:attribute}.

\newpage

%% file: sections/appendix_units.tex
This section describes how the programming framework~\cite{Hay2020-opendpprog} underlying the state-of-the-art \gls{dp} libraries~\cite{OpenDP2020-whitepaper, Berghel2022-tumult} can be generalized to support multiple privacy units.
The high-level idea of this \gls{dp} programming framework is to decompose a \gls{dp} algorithm into a \emph{transformation stage} and a \emph{measurement stage}.
In the \emph{transformation stage}, a chain of stable transformations is applied to the dataset, computing the algorithm's raw result (before adding noise).
Each transformation offers a stability guarantee: if the input distance between any two datasets is bounded by $d_{in}$, then the output distance is bounded by $d_{out}$.
The chain's overall sensitivity, defined as the maximum distance between any two neighboring datasets, can then be derived by composing these stability functions.
In the \emph{measurement stage}, a single \gls{dp} mechanism is applied to introduce the noise, and the noise is calibrated according to the sensitivity of the (chain of) transformations.
For example, in the case of the Gaussian mechanism, let $\triangle_2$ denote the $\ell_2$ sensitivity of a transformation chain $T(D)$.
The mechanism $M(D) = T(D) + x$, where $x \sim \mathcal{N}(0, \sigma^2)$ satisfies \gls{adp}~\cite{Dwork2014-dpbook} for:
\begin{equation}
\label{eq:gm}
\sigma^2 = \triangle_2  \cdot  \frac{2\log{1.25 / \delta}}{\epsilon^2}
\end{equation}

A challenge for \gls{dp} algorithms is that in many datasets, the number of contributions per user can be unbounded, leading to unbounded sensitivity and, consequently, infinite noise.
To address this, a contribution-bounding transformation can be applied to retain only $k$ contributions per user~\cite{Amin2019-dpcontrbiasvar}.
Contribution bounding is frequently implemented based on a privacy ID column in each record, ensuring that only $k$ records with the same privacy ID are retained~\cite{Berghel2022-tumult}.
The current programming framework assumes only a single privacy unit, with sensitivity calculated relative to that unit.
However, we can generalize this approach and calculate the sensitivity for different privacy units, such as user-level and user-month-level.
To realize this, we introduce a privacy ID for each privacy unit and apply a contribution-bounding transformation that limits the contributions at each level.
For example, we might enforce that each user-month privacy ID is limited to $k_M$ records, while each user privacy ID is limited to $k$ records.
In the measurement stage, the noise is calibrated with respect to the sensitivity $\triangle_2$ of the main privacy unit.
By rearranging the parameter calibration for $\epsilon$ and substituting the primary unit’s sensitivity ($\triangle_2$) with the auxiliary unit’s sensitivity ($\hat{\triangle}_2$), we obtain the privacy loss parameter $\hat{\epsilon}$ for the auxiliary privacy unit.
For example, for the Gaussian Mechanism:
\begin{equation*}
\hat{\epsilon} =  \sqrt{ \hat{\triangle}_2  \cdot  \frac{2\log{(1.25 / \delta)}}{\sigma^2}}
\end{equation*}

%% file: sections/appendix_proof.tex
\begin{proof}
    Assume, for contradiction, that there exist two \glspl{rule} $r_i \in R$ and $r_j \in R \setminus \{r_i\}$ such that $r_i \preceq r_j$ and $B_i \geq B_j$, but $r_i$ is constraining (i.e., not non-constraining).
    By \Cref{def:nonconstrain}, this implies that there exists a composition of mechanisms $M = (m_1, m_2, \ldots, m_N)$ that violates \gls{rule} $r_i$ but satisfies all rules in $R \setminus \{r_i\}$.
    Since $r_i \preceq r_j$, by \Cref{eq:rulerelation}, every mechanism matching the predicate of \gls{rule} $r_i$ must also match the predicate of \gls{rule} $r_j$, and the privacy unit of \gls{rule} $r_i$ is smaller or equal to that of $r_j$, i.e., $u_i \leq u_j$.
    As shown in \Cref{sec:policy:unit}, $u_i \leq u_j$ implies that privacy loss parameters in unit $u_j$ are at least as large as in unit $u_i$.
    Note that privacy parameters for all common \gls{dp} variants trivially define a partial order.

    Let $c_i$ and $c_j$ denote the cumulative privacy loss for the (composition of) mechanisms matching the predicates of \glspl{rule} $r_i$ and $r_j$, measured in units $u_i$ and $u_j$, respectively.
    Since the composition associated with $r_i$ is the same- or a sub-composition of that associated with $r_j$, privacy loss parameters are non-negative, and privacy loss parameters for unit $u_j$ are at least as large as for unit $u_i$, it follows that $c_j \geq c_i$.
    Moreover, because the composition does not violate \gls{rule} $r_j$, we have $c_j \leq B_j$.
    However, since the composition in unit $u_i$ violates \gls{rule} $r_i$, it must be that $c_i \nleq B_i$.
    Combining these, we obtain $B_j \geq c_j \geq c_i \nleq B_i$
    which implies $B_j \nleq B_i$, contradicting the assumption that $B_i \geq B_j$.
    Thus, our assumption that $r_i$ is constraining must be false, and $r_i$ is indeed non-constraining, as required.
    \end{proof}

%% file: sections/appendix_scopesampling.tex
Each request selects a set of attributes from the data schema
through a Bernoulli process with continuation probability $p_A = 0.75$ and attribute selection probabilities $[p(\sa_1),\ldots,p(\sa_{150})]$.
This selection process results in $k + 1$ attributes per request where $k \sim geometric(p_A)$, amounting to $5$ attributes per request on average.
We assume that attributes are non-uniformly selected by the requests by modeling the probability of attributes as a Zipf distribution, i.e.,
$p(\sa_i) = \frac{1}{i^{\alpha}} \quad \text{for} \quad i \in [1, 150]$
where the skewness parameter $\alpha$ equals $1$.
We assume the risk levels are equally distributed among the selection probabilities.
Each attribute in the data schema is associated with one or more categories, with membership level \textit{\gls{level:member}}, \textit{\gls{level:strong}}, or \textit{\gls{level:weak}}
that represent the attributes' dependency assumptions.
We assign attributes to categories through a second Bernoulli process with continuation probability $p_C = 0.6$ (i.e., $3.5$ categories per attribute on average) and category selection probabilities $[p(\sC_1),\ldots,p(\sC_{10})]$.
also sampled from a Zipf distribution with $\alpha=0.1$.
Each attribute is \gls{level:member} of its first sampled category, while additional category associations are classified as either
\gls{level:strong} or a \gls{level:weak} based on an equal split.